%% file: main.tex
\documentclass[10pt,journal,compsoc]{IEEEtran}
\usepackage{amsmath,amssymb,amsfonts}
\usepackage{graphicx}
\usepackage{textcomp}
\usepackage{float}
\usepackage{amsmath,amssymb,amsfonts}
\usepackage{graphicx}
\usepackage{textcomp}
\usepackage{booktabs} 
\usepackage{array}
\usepackage{url}
\usepackage{subfigure}
\usepackage{algorithmic}
\usepackage[linesnumbered,ruled,vlined]{algorithm2e}
\usepackage{changepage}
\usepackage{amsthm}
\usepackage{balance}
\usepackage{makecell}
\usepackage{mathptmx}
\usepackage{mathrsfs}
\usepackage{caption}
\usepackage{amsmath}

\allowdisplaybreaks[4]

\ifCLASSOPTIONcompsoc
  \usepackage[nocompress]{cite}
\else
  \usepackage{cite}
\fi
\ifCLASSINFOpdf
\else
\fi
\SetKwProg{Upon}{upon}{ do}{end}
\SetKwProg{Procedure}{procedure}{do}{}
\newtheorem{lemma}{\rm\textbf{Lemma}}

\newtheorem{definition}{\rm\textbf{Definition}}
\newtheorem{claim}{\rm\textbf{Claim}}
\begin{document}
%
\title{LinSBFT: Linear-Communication One-Step BFT Protocol for Public Blockchains}

\author{
    \IEEEauthorblockN{Xiaodong Qi$^1$,Yin Yang$^2$, Zhao Zhang$^1$, Cheqing Jin$^1$, Aoying Zhou$^1$} \\
    \IEEEauthorblockA{ $^1$East China Normal University,$^2$Hamad Bin Khalifa University} 
    xdqi@stu.ecnu.edu.cn,yyang@hbku.edu.qa,\{zhzhang,cqjin,ayzhou\}@dase.ecnu.edu.cn
}

\IEEEtitleabstractindextext{%
\input{abstract}

}

\maketitle

\IEEEdisplaynontitleabstractindextext

%
\IEEEpeerreviewmaketitle

\input{introduction}

\input{problem}
\input{LinSBFT}

\input{protocol}
\input{correct}

\input{experiment}
\input{relatedwork}
\input{conclusion}
{
\small
\bibliographystyle{abbrv}
\bibliography{reference}
}
\input{appendix}
\end{document}

%% file: abstract.tex
\begin{abstract}
This paper presents LinSBFT, a Byzantine Fault Tolerance (BFT) protocol with the capacity of processing over 2000 smart contract transactions per second in production. LinSBFT applies to a permissionless, public blockchain system, in which there is no public-key infrastructure, based on the classic PBFT with 4 improvements: (\romannumeral1) LinSBFT achieves $O(n)$ worst-case communication volume, in contract to PBFT's $O(n^4)$; (\romannumeral2) LinSBFT rotates the leader of protocol randomly to reduce the risk of denial-of-service attacks on  leader; and (\romannumeral3) each run of LinSBFT finalizes one block, which is robust against participants that are honest in one run of the protocol, and dishonest in another, and the set of participants is dynamic, which is update periodically. (\romannumeral4) LinSBFT helps the delayed nodes to catch up via a synchronization mechanism to promise the liveness.
Further, in the ordinary case, LinSBFT involves only a single round of voting instead of two in PBFT, which reduces both communication overhead and confirmation time, and employs the \emph{proof-of-stake} scheme to reward all participants.
Extensive experiments using data obtained from the Ethereum demonstrate that LinSBFT consistently and significantly outperforms existing in-production BFT protocols for blockchains. 
\end{abstract}

%% file: introduction.tex
\section{Introduction}
\IEEEPARstart{T}he blockchain technology, pioneered by Bitcoin \cite{GOOGLE:nakamoto2008bitcoin}, promises to revolutionize finance with a secure, decentralized and trustless protocol for processing transactions, which include money transfers and smart contracts. Many blockchain systems nowadays, however, suffer from poor scalability and slow confirmations, and consume vast amounts of energy \cite{GOOGLE:MaloneBitcoin}. At the heart of the problem is the widely-used proof-of-work (PoW) census mechanism, in which special power nodes, called miners, compete to solve cryptographic puzzles in order to gain the privilege of confirming transactions. Aside from scalability, latency and sustainability issues, PoW inevitably \emph{forks} 
\cite{DBLP:conf/sosp/Alogrand}; consequently, a confirmed transaction can still be reversed, though with diminishing probability as more confirmations arrive. In practice, applications usually wait for multiple confirmations, exacerbating the latency problem. 

Byzantine Fault Tolerance (BFT) protocols promise to solve the problems of PoW, since BFT involves negligible computations and guarantees no fork. However, classic BFT protocols scale poorly with the number of nodes in the network, due to their enormous communication cost. For example, PBFT, a textbook protocol, incurs worst-case $O(n^4)$ total transmissions for $n$ nodes \cite{DBLP:journals/corr/abs-1803-05069}. Consequently, most deployed BFT-based blockchains  support very few participants of the protocol (e.g., 21 in \cite{GOOGLE:EOS}), which can be elected delegates \cite{GOOGLE:EOS}, PoW winners \cite{DBLP:conf/sp/Kokoris-KogiasJ18}, or a random sample set obtained through cryptographic sortition \cite{GOOGLE:Dfinity, DBLP:conf/sosp/Alogrand}. As discussed in \cite{GOOGLE:Dfinity, DBLP:conf/sosp/Alogrand}, having fixed delegates defeats decentralization, and PoW introduces uncertainty in the participant set due to forks. Sortition-based BFT, on the other hand, only provides probabilistic guarantees on safety (i.e., no fork); further, its probability of failure is only small enough with a large sample (hundreds at least \cite{GOOGLE:Dfinity}), which might be already beyond the capacity of in-production BFT chains (e.g., 21 nodes in EOS \cite{GOOGLE:EOS}, and up to 16 in an earlier Hyperledger Fabric according to \cite{DBLP:conf/sigmod/BlockBench}). 

The state-of-the-art BFT protocol with deterministic safety guarantee is HotStuff (HS) \cite{DBLP:journals/corr/abs-1803-05069}, whose communication cost is $O(n)$ for each ``level", which roughly corresponds to a block. HS reduces one phase of voting in each level by pipeline of voting phases through multiple levels. In particular, an execution of the HS protocol covers multiple blocks, which amortizes costs. Each time a ``beacon'' proposes a block to nodes, who vote to finalize the block. HS improves the  worst  case communication  complexity  to $O(n^2)$,  using a combination of linear view change and threshold signatures.  While HS is attractive in theory, it is difficult to apply it to a public, permissionless blockchain, for four reasons. \textbf{\texttt{First}}, the beacon has too much power: it decides which transactions to include in the next block proposal and aggregates votes sent by others into a commit certificate (CC). In \cite{DBLP:journals/corr/abs-1803-05069}, the authors suggest using PoW to implement the beacon, which leads to forks and contention, complicating system design. Besides,  a malicious beacon may send CC to part of participants selectively resulting in partition of all participants. Consequently different nodes may run at different block height, which influences the liveness of protocol. \textbf{\texttt{Second}}, there can be a cascading sequence of $f$ fault beacons, leading to $f+1=O(n)$ levels in HS. If HS changes the beacon in round robin manner instead of PoW which avoids forks, adversary can attack HS with  consecutive $f$ beacons easily. This will increase the communication cost for a block and latency of transactions sent by users. \textbf{\texttt{Third}}, HS depends on the existence of a centralized public-key infrastructure (PKI) to generate keys for threshold signature, which is not practical in a permissionless, public blockchain setting. Besides, HS lacks a sufficient incentive mechanism in open setting. If the rewards are owned by the leader, others may attack the leader  deliberately to gain more benefits.

\textbf{\texttt{Last}}, the protocol spans over multiple blocks, which is a problem, because to guarantee safety, BFT requires that at least $2/3n$ nodes be honest throughout the protocol. In reality, a node could be honest for one block, and dishonest for another; further, in a public chain the set of nodes can also change. For example, there are four participants $A$, $B$, $C$ and $D$, among which $A$, $B$ and $C$ are honest and $D$ is malicious at the beginning. The malicious participant $D$ can help $A$ and $B$ to construct a quorum, which can involve two phases of voting for a block $X$. The proposal for block $X$ is not received by participant $C$.  Assume that just $A$ commits block $X$, and $B$ does not receive the CC from the beacon due to the asynchrony of network.  At this time, $B$ becomes malicious and $D$ becomes honest, where there is only a malicious participant in HS as well. Then $B$ can vote for any block $Y$ conflicting with $X$, and participant $B$, $C$ and $D$ construct another quorum which can commit block $Y$. However, honest participant $A$ and $C$ commit conflict blocks, which compromises the safety of protocol. 

This paper proposes LinSBFT, which achieves amortized $O(n)$ total transmissions under deterministic safety, involves no PoW module, in a public setting. Meanwhile, LinSBFT does not require a public-key infrastructure (PKI), and is compatible to the \emph{Proof-of-Stake} (\emph{PoS}) scheme commonly used in public blockchains, which defends against \emph{Sybil} attacks. Specifically, LinSBFT is based on PBFT with four key improvements as follows.  

\begin{itemize}
    \item { LinSBFT reduces communication costs of consensus for each block height with three key techniques, each by $O(n)$: linear view change \cite{DBLP:journals/corr/abs-1803-05069}, threshold signatures \cite{GOOGLE:Dfinity} and leader selection via verifiable random functions (VRFs) \cite{DBLP:conf/sosp/Alogrand}.
    }
    \item {The proposed protocol reduces the risk of denial-of-service (DoS) attacks on the leader in public setting by changing leader randomly(i.e., block proposer) for every block. With random leader rotation,
    the adversary is infeasible to predict next leader in advance.
    }
    \item{LinSBFT is against nodes with changing honesty as well as a dynamic node set. LinSBFT guarantees the safety of protocol by an novel locking mechanism even though the honesty of a participant is allowed to change from block to block. Instead of static participant set, LinSBFT allows nodes to join and leave the protocol periodically, where the time is divided into epochs and participant set is update at the beginning of each epoch.
    }
    \item{To deal with situation that malicious leader partitions all nodes, LinSBFT designs a synchronization mechanism to help delayed nodes catch up without increasing the communication complexity, which promises the liveness of protocol. 
    }
\end{itemize}

In addition, similar to HS,  in the ordinary case with a non-faulty leader and synchronous network, LinSBFT reduces communication costs and block confirmation time by piggybacking the ``Commit'' vote for the previous block onto the ``Prepare'' vote for the current block, which we elaborate in Section \ref{sec:LinSBFT}. The major contributions are summarized below:

\begin{enumerate}
    \item {LinSBFT achieves $O(n)$ worst-case communication volume, with three key techniques: linear view change, threshold signatures and leader selection via VRFs, which avoids $f$ leader failure and DoS attacks of leader.}
	\item{Each run of LinSBFT finalizes one block with strict safety (no fork), which is robust against changeable honesty and dynamic node set. Besides, a synchronization mechanism is applied to help delayed participant to catch up  for the liveness.}
	\item{We give a formal proof of the correctness of safety, liveness and linear complexity for LinSBFT.}
	\item {An implementation of LinSBFT  and extensive experiments, based on real data from the Ethereum, demonstrate that LinSBFT consistently and significantly outperforms existing in-production blockchain BFT protocols.}
\end{enumerate}

The rest of this paper is structured as follows. Section \ref{sec:problem} provides necessary background and explains the problem setting. Section \ref{sec:LinSBFT} presents the major components of LinSBFT, and Section \ref{sec:protocol} details the complete protocol. Section \ref{subsec:correctness} proves the security and performance guarantees of LinSBFT. Section \ref{sec:evaulation} contains a thorough set of experiments. Section \ref{sec:rw} reviews related work.  Finally, Section \ref{sec:conclusion} concludes the paper with directions for future work.

%% file: problem.tex
\section{PRELIMINARIES and PROBLEM SETTING}\label{sec:problem}
\subsection{Preliminaries} \label{subsec:preliminaries}
\textbf{Threshold signature}. An $(n, t)$-$ts$ \emph{threshold signature} on a message $m$ is a single, constant-sized aggregate signature that passes verification if and only if at least $t$ out of the $n$ participants sign $m$. Note that the verifier  does not need to know the identities of the $t$ signers.
Without a threshold signature scheme, the verifier has to either receive and verify $t$ individual signatures, which requires $O(n)$ transmissions when $t=O(n)$. Threshold signature brings down this cost to $O(1)$. LinSBFT employs a popular implementation of threshold signature based on the BLS signature scheme \cite{GOOGLE:SurveySignAggr}. However, \emph{threshold signature} is not free lunch, a fact that is sometimes ignored in the literature. In particular, a threshold signature scheme requires special, correlated public/private key pairs. Generating such key pairs in a decentralized setting requires a distributed key generation (DKG) protocol, which is communication-heavy. For example, the Joint-Feldman algorithm \cite{DBLP:conf/eurocrypt/GennaroJKR99}, e.g., used in Dfinity \cite{GOOGLE:Dfinity}, incurs $O(n^3)$ network transmissions to broadcast the coefficients of an order-$t$ polynomial when $t=2f+1=O(n)$. 
LinSBFT employs the DKG solution in \cite{DBLP:conf/eurocrypt/CannyS04}, which requires $O(n\, polylog\, n)$ communications, and provides a probabilistic guarantee on the correctness of the generated key pairs(i.e., threshold signatures can be successfully created with these keys), where the probability of failure can be made arbitrarily small, e.g., below  $10^{-18}$.

\textbf{Verifiable hash function $(VRF)$}. A $VRF$ is a pseudo-random generator whose output is verifiable (i.e., on whether a given number is indeed the output of the VRF), random, uniformly distributed, and unpredictable beforehand. 
A simple $VRF$ (e.g., used in Algorand \cite{DBLP:conf/sosp/Alogrand}) under the \emph{random oracle} model (i.e., there exists an \emph{ideal hash function} $H$ whose outputs are random and uniform), is $H(s)$, where $H$ is the ideal hash function uniqueness\footnote{In practice, $H$ can be approximated with a cryptographical one, such as SHA-3 \cite{GOOGEL:SHA3}.} and $s$ is a signature that satisfies, i.e., there is a unique signature for a given message and a private key. The BLS signature scheme \cite{GOOGLE:SurveySignAggr}, for instance, satisfies this property. The output of $H(s)$ is clearly random and uniform, due to the random oracle assumption. Meanwhile, given the source message $m$ (from which $s$ is obtained) and $s$'s corresponding public key, one can verify that a given value is indeed the result of $H(s)$. Further, the function's output is unpredictable beforehand, since it is infeasible to obtain $s$ without knowledge of its secret key. 



\textbf{PBFT in blockchain}. A classic BFT protocol is PBFT \cite{DBLP:conf/osdi/PBFT}. Assume that there are $n \ge 3f+1$ nodes in total, among which $f=\lfloor\frac{n-1}{3}\rfloor$ are malicious. PBFT involves three steps to reach consensus, \emph{Pre-Prepare}, \emph{Prepare} and \emph{Commit} in the ordinary case that the leader is not faulty, and the network is synchronous. PBFT can be used to obtain consensus on a block in a blockchain setting (e.g., implemented in \cite{buchman2016tendermint}), as follows.  The protocol involves one or more rounds, each of which has a leader, which can be chosen, e.g., in a round-robin manner. The \emph{leader} in the $v$-th round at height $l$ (denoted by $L_{l,v}$) proposes a block $B_{l,v}$, and broadcasts its hash value $H_{l,v}=H(B_{l,v})$ in a \emph{Pre-Prepare} message. Upon receiving such a \emph{Pre-Prepare}, a node responds by broadcasting \emph{Prepare} messages on $H_{l,v}$. Once a node receives $n-f$ \emph{Prepare}'s on $H_{l,v}$, it assembles them into a \emph{prepared certificate} $(PC)$, and broadcasts \emph{Commit} messages about the $PC$. A node who receives $n-f$ \emph{Commit} on the $PC$ of $H_{l,v}$ assembles them into a \emph{committed certificate} $CC$ and is ready to finalize block $B_{l,v}$, after verifying all transactions in the block. 

When consensus cannot be reached in a round $v$ at height $l$ within a given timespan, the protocol enters a new round $v+1$ with a different leader $L_{l,v+1}$, which is called a \emph{view change}. Possible causes for a view change include a malicious leader, or message losses due to network failures. The original PBFT protocol involves $O(n^3)$ transmissions per view change. In the worst case, there can be $f=O(n)$ faulty leaders, leading to a total transmission cost of $O(n^4)$. To ensure the safety of protocol i.e., honest nodes never finalize conflicting blocks at same height, PBFT requires all nodes stay static where a honest node cannot become malicious. This assumption is unrealistic in an open setting since each node always chases the maximization of benefits by various means. Besides,  PBFT changes the leader (i.e., block proposer) in a round robin manner, which suffers from denial-of-service attacks on the leader. This attack may compromise the liveness of protocol.

\textbf{Linear View Change}. A recent algorithm called \emph{linear view change} $(LVC)$ reduces the cost to $O(n^2)$ \cite{DBLP:journals/corr/abs-1803-05069}. At the beginning of a new round $v+1$, each node sends a \emph{NewView} message to $L_{l,v+1}$, along with the prepared certificate $PC$. $L_{l,v+1}$ then broadcasts the $PC$ with the highest round number among all collected $PC$s, along with the hash $H_{l, v+1}$ of its proposed block $B_{l,v+1}$. Since a $PC$ contains $n-f$ \emph{Prepare} messages, it has size $O(n)$; thus, collecting and broadcasting a $PC$ cost $O(n^2)$ transmissions. Note that each node does not need to send  $PC$s for previous blocks to $L_{l,v+1}$ because all blocks are chained by hash values in block headers and the $PC$ for the latest block can represent the entire blockchain.

Via threshold signature, the size of certificates $PC$ and $CC$ can be reduced to constant. 
PBFT with linear communication applies a standard trick (e.g., used in \cite{DBLP:journals/corr/SBFT}\cite{DBLP:journals/tocs/Zyzzyva}) with a \emph{collector} that collects and aggregates \emph{Prepare} and \emph{Commit} votes from all validators. A natural choice for the collector is the leader $L_{l,v}$. In the ordinary case, leader $L_{1,v}$ collects \emph{Prepare} votes  from all validators on $H_{1,v}$. Then, $L_{1,v}$ derives an $(n, t =n-f)$-threshold signature $ts(H_{l,v})$ on $H_{l,v}$ from these \emph{Prepare} messages, and creates a prepared certificate $PC_{l,v}$ containing $ts(H_{l,v})$. After that, $L_{l,v}$ broadcasts $PC_{l,v}$ to all validators, each of which responds with a \emph{Commit} message with its signature on $PC_{l,v}$. Again, the  $L_{l,v}$ plays the roles to collect \emph{Commit} messages and derives an $(n, n-f)$-$ts$ threshold signature $ts(PC_{l,v})$, and makes a committed certificate $CC_{l,v}$ including $ts(PC_{l,v})$. Then $L_{l,v}$ broadcasts $CC_{l,v}$ to all validators. Once a validator receives $CC_{l,v}$, it finalizes block $B_{l,v}$. This completes the protocol.

\subsection{Problem Settings}
\textbf{Adversary model}. Different from PBFT,  out of the $n$ participants of LinSBFT numbered from 0 to $n-1$, at most $f$ (such as $n\ge 4f+1$) can be malicious for each execution of the protocol, which can misbehave in arbitrary ways. The remaining participants are honest, who strictly follow the protocol.  Since the malicious nodes can collude, one can view them as corrupted and controlled by a single mastermind, referred to as \emph{the adversary}. 
Instead, LinSBFT assumes that the adversary is static and rushing (i.e., it has to choose the $f$ nodes to corrupt before a protocol run), rather than adaptive (which can instantly compromise any node at any time). Further, the adversary is assumed to take some time (a constant) to compromise nodes, which is detailed further in Section \ref{subsec:viewchange}.

Since the blockchain setting contains an infinite number of transactions split into blocks, it is important to clarify the timespan over which $f$ is defined. In LinSBFT, the honesty of a participant is allowed to change from block to block, e.g., an honest participant may become malicious when it encounters a transaction that strongly motivates it to cheat. Accordingly, over multiple blocks, it is possible that every participant is malicious at some point, and yet the system remains secure as long as $n\geq4f+1$ holds for each block. Meanwhile, the changing honesty assumption makes it necessary to place the condition that the adversary cannot carry over knowledge of private keys of malicious nodes from one execution of the protocol to another. 
Compared to PBFT, the adversary model fo LinSBFT is more realistic for public blockchains.
A BFT protocol must satisfy safety (i.e., no fork) and liveness (the protocol eventually terminates) when the number of malicious nodes is less than a super-majority of all participants. Specifically, in a partially synchronous network, LinSBFT satisfies both safety and liveness deterministically, with zero chance of failure. Further, LinSBFT also satisfies deterministic safety even when the network is asynchronous, i.e., it never forks.

\textbf{Communication model}. Following common practice in the literature, LinSBFT assumes that the network is partially synchronous \cite{DBLP:journals/jacm/DworkLS88}, i.e., after an unknown future \emph{Global Stabilization Time} (GST), any message between two honest nodes is delivered within $\varDelta$ time. This model also captures the more common situation that periods of synchrony and asynchrony interleave, and there are sufficiently long periods of synchrony that allow the protocol to finish \cite{DBLP:journals/corr/abs-1803-05069}. It is worth pointing out that the parameter $\varDelta$ above does not take into account the network topology, or the total amount of network traffic. For instance, an all-to-all broadcast, which clearly involves $\varOmega(n^2)$ messages, can be said to take $O(n)$ time, since every node sends/receives $n-1$ messages, each of which takes constant time, i.e., up to $\varDelta$. This analysis is not valid when the network has a saturated critical link, e.g., when half of the nodes reside in America and the other half in China. In this case, since $\varOmega(n^2)$ traffic pass through this critical link with limited bandwidth, the total time is no less than $\varOmega(n^2)$. For this reason, LinSBFT focuses on minimizing communication volume rather than time.


%% file: LinSBFT.tex
\section{LinSBFT} \label{sec:LinSBFT}

\subsection{Overview} \label{subsec:overview}
A blockchain consists of blocks linked through the \emph{PreHash} attribute at each block, which is the hash value of the previous block, called its \emph{parent}. For a given block, the set of all previous blocks are its \emph{ancestors}. This linked structure indicates that when a node votes for a particular block, it implicitly votes for all its ancestors \cite{DBLP:journals/corr/abs-1803-05069, DBLP:journals/Casper} as well. In the following, we use the terms ``node'' and ``validator'' interchangeably. At the beginning, we give an overview of LinSBFT for better understanding.

\textbf{Ordinary case.} In the ordinary case, LinSBFT finishes in a single round at each block height, in two steps: \emph{Propose} and \emph{Vote}. Validators only vote once in the \emph{Vote} step, and there is no explicit \emph{Commit} steps as in PBFT (Section \ref{subsec:preliminaries}). LinSBFT finishes consensus height by height, and in each block height, it may involve multiple rounds to reach consensus among validators. Specifically, let ${R}_{l,v}$ represent the $v$-th round at height $l$. In the \emph{Propose} step, the leader for round ${R}_{l,0}$, denoted as $L_{l, 0}$,
chooses a batch of unconfirmed transactions from its local transaction pool to compose a block $B_{l,0}$, and broadcasts a signed proposal  $P_{l,0}$. Upon receiving $P_{l,0}$, each validator enters round ${R}_{l,0}$. Then, in the \emph{Vote} step, each validator sends vote $V_{l,0}$  for the proposed block to a collector $C_{l,0}$, which is also the leader for the next block height, i.e., $C_{l,0} = L_{l+1,0}$. The collector node, upon receiving votes from at least $n-f$ validators, derives an $(n, n-f)$-$ts$ threshold signature $ts(H_{l,0})$, and creates a vote certificate $Cert_{l,0}$ as follow:
$$Cert_{l,v}=\left<l,v,H_{l,v},ts(H_{l,v})\right>$$
where $v$ is the round number, which is 0 in our description so far.

\begin{figure}[t]
	\setlength{\abovecaptionskip}{0.1cm}
	\center{\includegraphics[width=8.5cm]  {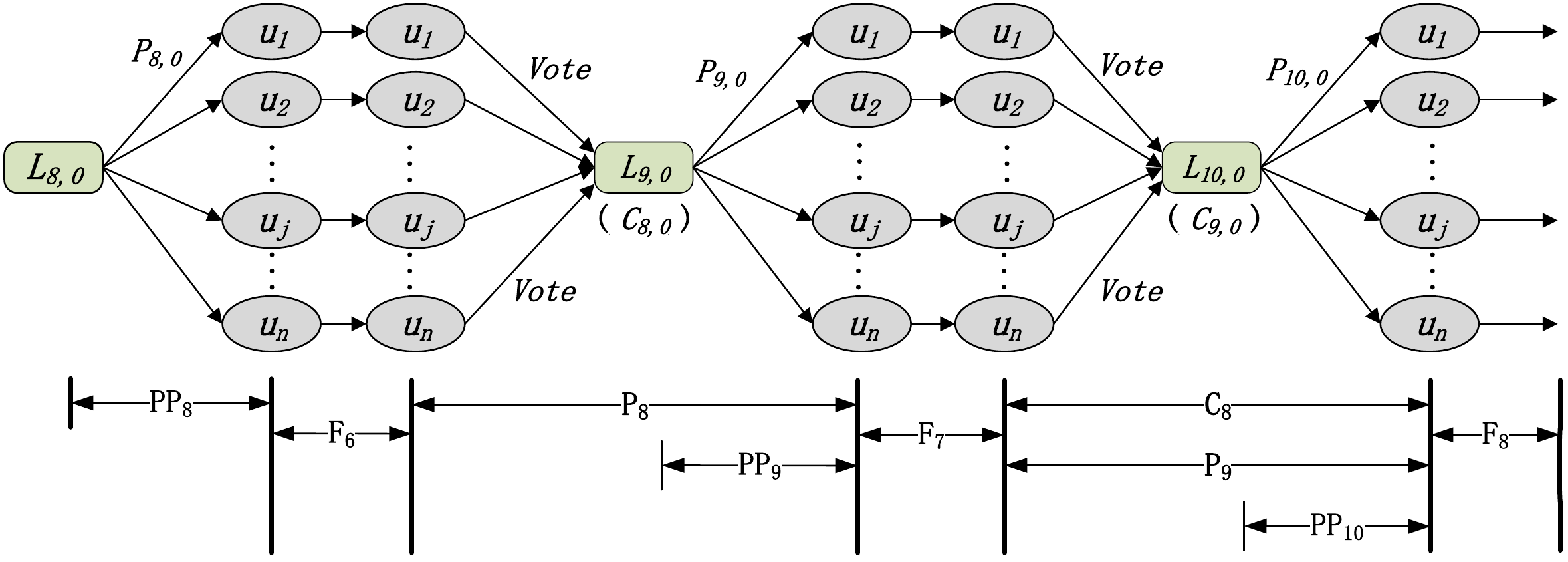}}        
	\caption{\label{fig:LinSBFT} Example of the ordinary case of LinSBFT. For each block height $l$, its first selected leader $L_{l, 0}$ broadcasts \emph{Propose} messages, which are voted by the validators. These votes are collected by the collector $C_{l, 0} = L_{l+1, 0}$. The equivalent \emph{Pre-prepare}, \emph{Prepare} and \emph{Commit} phases in PBFT are annotated as $PP_l$, $P_l$ and $C_l$, respectively. $F_l$ (and the corresponding arrows between nodes to themselves) indicates that validators \emph{Finalize} the block at height $l$.}
\end{figure}

The size of the above vote certificate is constant, since a threshold signature is of constant size (refer to Section \ref{subsec:preliminaries}). After that, $C_{l,0}$ (i.e., $L_{l+1,0}$) enters the \emph{Propose} step for height $l+1$, and broadcasts $P_{l+1,0}$ containing $Cert_{l,0}$. Each validator verifies the correctness of $P_{l+1,0}$ after receiving it from $C_{l,0}$. Similarly, each validator enters round ${R}_{l+1,0}$ and sends vote $V_{l+1,0}$ for $P_{l+1,0}$ to collector $C_{l+1, 0}$ in \emph{Vote} step. Collector $C_{l+1, 0}$ deals with these votes similarly as $C_{l,0}$ at block height $l$, and broadcasts the proposal $P_{l+2, 0}$ for the next block height $l+2$. Upon receiving $P_{l+2, 0}$ from $C_{l+1,0}$, each validator finalizes $B_{l,0}$ after verifying $Cert_{l+1,0}$. 

We emphasize that the process to finalize a block, e.g., $B_{l, 0}$, still requires 2 phases of voting. Compared to PBFT, LinSBFT does not simply eliminates a round of voting; instead, it \emph{pipelines} the voting phases for adjacent blocks. For instance, in Figure \ref{fig:LinSBFT}, for block $B_{8,0}$, the \emph{Vote} step at height $8$ plus the \emph{Propose} step at height $9$ is equivalent to the \emph{Prepare} step in PBFT, and the \emph{Vote} step at height $9$ plus the \emph{Propose} step at height $10$ is equivalent to \emph{Commit} phase in PBFT. A validator finalizes $B_{8,0}$ when it receives vote certificate $Cert_{9,0}$ along with $P_{10,0}$. Essentially, a vote in LinSBFT signifies both the \emph{prepare-vote} for current height and the \emph{commit-vote} for the last height. Besides, the collector of each round varies, which is determined by VRF detailed in next subsection.

\textbf{View change.} For each round ${R}_{l,v}$ of the protocol, there is unique collector $C_{l, v}$, which is determined by a VRF, explained later. If validators fail to reach consensus for a proposal, the view change subprotocol is triggered, and validators move to a new round. A major challenge in LinSBFT is that validators may run at different block heights, e.g., due to network partitioning or malicious collectors/leaders. For example, in Figure \ref{fig:partition}, there are two consecutive faulty leaders $L_{9,0}=u_2$ and $L_{10, 0}=u_3$. The former (i.e., $u_2$) sends out proposal $P_{9,0}$ only to nodes $u_1$-$u_{3f+1}$, but not the rest. Similarly, $u_3$ selectively sends out proposal $P_{10, 0}$ to nodes $u_{f+1}$-$u_{3f+1}$. Consequently, different nodes now run at different heights, e.g., node $u_n$ is still at height 8 as it has not received $P_{9, 0}$. 

To tackle this problem, LinSBFT follows a solution similar to BFT-SMART \cite{DBLP:conf/edcc/BFT-smart}, in which each validator participates not only in the protocol for the current height $l$, but also the previous height $l-1$. In particular, all honest validators running at height $l$ keep functioning (i.e., voting and proposing) for rounds at height $l-1$. This is needed since the protocol can enter height $l$ with $n-2f$ votes; consequently, there may be up to $n-3f-1$ honest validators left behind at height $l-1$. Further, when a validator runs at height $l$, it only proposes and votes for blocks at heights $l$ and $l-1$, which is important to guarantee safety, as we show in Section \ref{subsec:lock}. Whenever a validator receive messages for a higher height (e.g., $u_n$ in Figure \ref{fig:partition}), it realizes that it has fallen behind others, and switches to \emph{synchronization mode}, in which it actively requests new blocks from its peers until it catches up with the current block height, which is detailed in Section \ref{subsec:future_vote}.

\begin{figure}[t]
	\setlength{\abovecaptionskip}{0.2cm}
	\center{\includegraphics[width=8.5cm]  {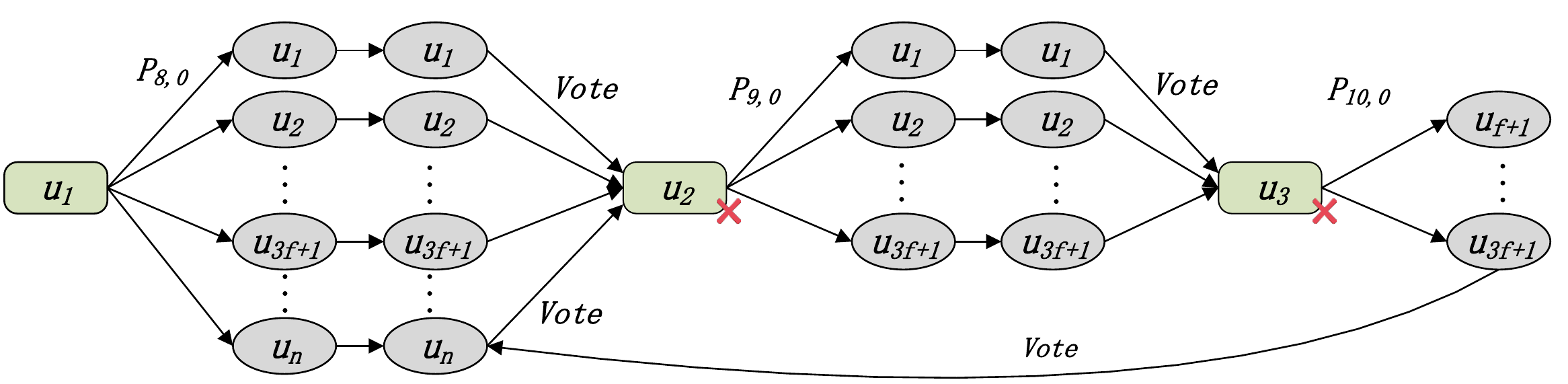}} 
	\caption{\label{fig:partition}Example of different validators at different heights, due to two consecutive faulty leaders $L_{9, 0}=u_2$ and $L_{10, 0}=u_3$, who selectively send block proposals to some but not all validators. Assume that after $u_3$, the next collector is $C_{10, 0}=u_n$. Then, $u_n$ is at height 8, but receives votes for a block at height 10.}
\end{figure}

\textbf{Messages.} A message sent by a validator is defined in form: $${message} \, m=\left<Type, Value\right>_{s} $$ in which $s$ is the sender of messages, $Type$ is the type of message which can be \emph{PROPOSE} or \emph{VOTE}, and $Value$ is actual value to be sent as shown in Equation (\ref{equation:Value}). 
\begin{subequations} \label{equation:Value}
\small
	\begin{align} 
	&P_{l, v}=\left<l, v, B_{l,v}, Cert,\sigma\left(Cert\right)\right> \label{subequation:proposal}\\
	&V_{l,v}=\left<l, v, \sigma\left(H_{l, v}\right), Cert\right> \label{subequation:vote}
	\end{align}
\end{subequations}
where $P$ represents \emph{Proposal} and $V$ denotes \emph{Vote}. Each message must be signed by the corresponding sender in order to establish authenticity. For brevity,  message signatures are omitted in the above equations.

In Equation (\ref{subequation:proposal}), $Cert$ is the highest vote certificate at the collector and  $\sigma\left(Cert\right)$ is the signature of collector $C_{l-1,v'}$ (i.e., $L_{l, v}$), which forces the collector to derive a valid threshold signature; otherwise, $\sigma$ $\left(Cert\right)$ would be the evidence of collector's misbehavior, and the collector would be slashed. A $Cert$ is the highest for a validator if and only if it has the highest height, breaking ties by round number. In Equation \ref{subequation:vote}, $\sigma\left(H_{l, v}\right)$ is the signature for proposed block hash and $Cert$ is the highest vote certificate owned by the sender as well.

\subsection{View Change} \label{subsec:viewchange}

Similar to PBFT, the view change subprotocol of LinSBFT is triggered when the nodes cannot reach consensus in a single round. This can be due to an asynchronous network (e.g., when more than $n/4$ nodes are offline), or the presence of malicious collectors/leaders. Specifically, a faulty leader may:  (\romannumeral1) propose multiple blocks or an invalid block; (\romannumeral2) remain silent indefinitely; (\romannumeral3) send valid proposal to some of the validators, but not to the rest.

For the first case, the faulty leader caught cheating is slashed. For the remaining two cases, the validators cannot distinguish whether they do not receive a proposal due to a faulty leader or network asynchrony. To ensure liveness, each validator sets a timer for every round. When the timer expires, a \emph{view change} is triggered, and the protocol enters a new round, say, round ${R}_{l,v+1}$. If a collector for round $R_{l,v}$ cannot collect $n-f$ votes within a timespan, it proposes a new block for $R_{l,v+1}$.  Similar to HotStuff, LinSBFT handles a view change with the Linear View Change (LVC) algorithm \cite{DBLP:journals/corr/abs-1803-05069}. The essence of LVC is that the leader of the next round sends its highest vote certificate instead of all vote certificates, which reduces transmission volume during a view change by a factor of $O(n)$.

\textbf{Random collector selection}. In all previous protocols based on PBFT, there can be a cascading sequence of $f$ fault leaders, leading to $f+1=O(n)$ rounds. LinSBFT avoids this situation by selecting collectors (leaders) randomly, using a VRF (refer to Section \ref{subsec:preliminaries}). With random collectors, the probability of having a sequence of malicious collectors diminishes exponentially with the length of the sequence. Specifically, with $f<n/4$ malicious validators, having a sequence of $I$ malicious collectors has probability smaller than $(1/4)^I$. Therefore, the probability that the next collector is malicious becomes negligible (i.e., smaller than a given $\rho$) after a constant number $(I >-\log_4 \rho)$ of collector changes. In practice, a common choice of $\rho$ is $10^{-18}$, whose inverse is larger than the total number of seconds since the beginning of the universe \cite{DBLP:conf/sosp/Alogrand}.  


Random collector selection requires a common source of randomness among the validators. In LinSBFT, this is provided by a VRF on the vote certificate for the previous block at height $l-2$. In particular, we have:
\begin{equation} \label{equation:collector}
\small
C_{l,v} = H\left(Cert_{l-2,v'}|v\right)\rm {mod\;} \emph{n}, 
\end{equation}
where ``$|$'' denotes concatenation, and $Cert_{l-2,v'}$ is the vote certificate at the previous block height $l-2$ held by validator $u$ since each validator only has a vote certificate for previous block height $l-2$ which is discussed in Section \ref{subsec:lock}, and $Cert_{l-2,v'}$ is the vote certificate indicating that $n-f$ participants have voted for previous block at height $l-2$. In LinSBFT, $Cert_{l-2,v'}$  is first generated by the collector of the last round at height $l-2$. Based on the assumption that the adversary takes time to corrupt validators, which equals the duration from the creation of $Cert_{l-2,v'}$ at the last collector at height $l-2$ and the beginning of the protocol at height $l$, the output of the VRF above is unpredictable to the adversary beforehand. Meanwhile, the VRF is clearly known to all validators that have finalized at height $l$, which is a necessary condition for entering the protocol for height $l$. This manner also avoids denial-of-service attacks on the collector(leader), which threatens the liveness of protocol, since the collector is unpredictable. An example is shown in Figure \ref{fig:view_change}.
\begin{figure}
	\center{\includegraphics[width=8.7cm]  {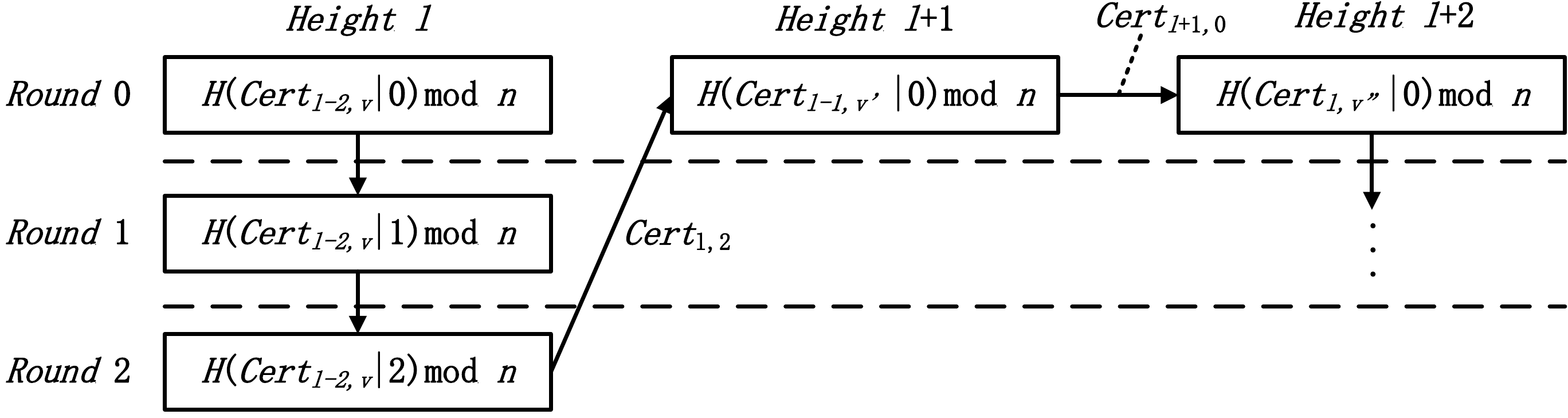}} 
	\caption{\label{fig:view_change} Example of random collector (leader) selection for three consecutive blocks.} 
\end{figure}
With random collectors, the number of view changes in a synchronized network becomes $O(1)$, unless with negligible
probability. Since each view change takes $O(n)$ transmissions, the total communication cost for all view changes is still $O(n)$, unless with negligible probability.


\subsection{Block Locking Mechanism} \label{subsec:lock}
First we explain the necessity of a block locking mechanism with an example. Consider a setting with five validators, $u_1$-$u_5$, who all voted on a block proposal $P_{l, v}$ at height $l$ and round ${R}_{l,v}$. The collector for this round then assembles a vote certificate $Cert_{l, v}$ and broadcasts it. Now, suppose that $u_1$ fails to receive $Cert_{l, v}$, and the remaining validators $u_2$-$u_5$ receive it and continue to vote on another proposal $P_{l+1, 0}$ whose block is linked to the one in $P_{l, v}$. Suppose that $u_2$-$u_5$ all vote on $P_{l+1, 0}$, leading to vote certificate $Cert_{l+1, 0}$. This time, $u_5$ is the only one who receives $Cert_{l+1, 0}$, whereas $u_2$, $u_3$ and $v_4$ fail to receive it. Then, according to the protocol, $u_5$ finalizes the block contained in $P_{l, v}$. The other three users, i.e., $u_1$-$u_4$, eventually time out for round ${R}_{l, v}$, and enters a new round ${R}_{l,v+1}$. Since $u_1$-$u_4$ form a super-majority of the 5 validators, they could reach consensus on a different block than the one in $P_{l, v}$, leading to a fork as $u_5$ has already finalized $P_{l, v}$. In LinSBFT, this problem is addressed using a locking mechanism. Note that locking in LinSBFT is necessarily more complicated than other PBFT-based protocols, due to the use of VRF-based leader selection, which chooses different leaders at different heights that could be deadlocked in an incorrect protocol. 

Specifically,  we introduce two basic concepts, \emph{conflict} and \emph{comparability}. Two blocks $B$ and $B'$ are \emph{conflicting} if one is not an ancestor of the other according to the \emph{PreHash} links, donated by $B\vdash\!B'$. We say a vote certificate $Cert$ for block $B'$ conflicts with other block $B$ when $B'\vdash B$, denoted by $Cert\vdash\!B$. If blocks $B$ and $B'$ have identical \emph{PreHash} and $B$ is proposed in a round with larger round number than $B'$, we say $B$ is \emph{larger} than $B'$, denoted by $B>B'$. Otherwise, if $B$ and $B'$ do not share the same parent block, they are called \emph{incomparable}. 

For each validator, there are two types of locks, \emph{Propose-lock} and \emph{Vote-lock}. A propose-lock $PLock(B,l)$ (or vote-lock $VLock(B,l)$) contains a block $B$, which means its corresponding validator is currently locked on block $B$ at height $l$. In general, when a validator has a propose-lock $PLock(B,l)$, it can only propose the locked block $B$ when it becomes the proposer. Similarly, a validator only can vote for the locked block $B$ at height $l$ when it has a vote-lock $VLock(B,l)$.

\textbf{Propose-Lock.} If a validator runs in round ${R}_{l-1,v}$, it is locked when it receives a proposal $P_{l,0}$ containing block $B$ with valid vote certificate $Cert_{l-1,v}$. Then it move to height $l$ and has a propose-lock $PLock(B,l)$. A validator only propose locked block if it is the proposer for a  round at height $l$, but it can vote for any valid block.

\textbf{Propose-Unlock.} A validator has $PLock(B,l)$. It may release a lock after seeing a $Cert$ for block $B'$ such that $B'>B$ . Then validator is relocked on $B'$. If validator has no \emph{Propose-Lock} at height $l$, it is locked on $B'$ as well.

\textbf{Vote-Lock.} A validator is locked when it receives a valid $Cert$ for block $B$ at height $l$ and hasn't voted for any block $B'$ such that $B'>B$ . Then it has a vote-lock $VLock(B,l)$. Validator only votes for locked block proposed at height $l$.  This prevents validators from voting for block, and then contributing to another vote certificate for a conflicted block in next rounds, thereby compromising safety.

\textbf{Vote-Unlock.} A validator having $VLock(B,l)$ may only release a lock after seeing a $Cert$ for a block $B'$ such that $B'>B$ if it hasn't voted for $B'$. This allows validators to unlock if they vote something the rest of the network doesn't want to finalize, thereby protecting liveness, but does it in a way that does not compromise safety, by only allowing unlocking if there has been a vote certificate in a round after that in which the validator became locked.

The \emph{Propose-Lock} increases the probability of finalization of the block that there are validators locked on. For example, without \emph{Propose-Lock}, fault collector only sends proposal with a vote certificate for block $B$ to one honest validator $u$ when $n=4f+1$. Then $u$ is locked on $B$ and $f$ faulty validators all keep silent. Honest validators change views and propose new block, but never get a valid vote certificate due to lack of the locked validator's vote. Once the locked validator becomes the collector, it proposes $B$ again and a valid certificate may be derived. However, it will take $O(f)$ view changes to move a round whose collector is the locked validator theoretically. Through \emph{Propose-Lock}, honest collector proposes the propose-locked block which avoids $O(n)$ round changes. Furthermore, since an honest node at a round may become malicious at another round, thereby a locked node may violate the locking mechanism after the change. However, if a validator is honest and accepts a block $B$ at a round, it must accept all $B$'s ancestors and locked on them correctly. Here, we give some properties of \emph{Propose-Lock} and \emph{Vote-Lock}.

\begin{lemma}\label{lemma:lock}
	If an honest validator has $PLock(B,l)$ and $VLock(B', l)$ at the same time, then $B$ must be identical to $B'$.
\end{lemma}
\begin{proof}
	Block $B$ and $B'$ have same parent block, hence they are comparable. By assumption, if $B>B'$, there must be a vote certificate $cert$ for $B$.  If validator has $VLock(B', l)$ yet,  it vote-unlocks when receives $cert$ according to the rules of lock and has $VLock(B,l)$. Otherwise,  If $B<B'$, there must be a vote certificate $cert'$ for $B'$. Then validator  propose-unlocks and has $PLock(B',l)$. Hence, $B$ must be identical to  $B'$.
\end{proof}

\begin{lemma} \label{lemma:unlock}
	If a validator $u$ runs at height $l$, the lock $PLock(B,l-2)$ and $VLock(B',l-2)$ will never be unlocked any more if  it keeps honest even validators are rushing.
\end{lemma}

\begin{proof}
	Due to the condition $n\ge 4f+1$, the number of all validators $n$ can be represented in form $n=4f+k$ ($1\le k \le 4$). By assumption, because validator runs at height $l$, it must has received a vote certificate $cert$ for block $blk\nvdash B'$ at height $l-1$. It means that more than $n-f=3f+k$ validators has $VLock(B',l-2)$, and there are at least $n-2f=2f+k$ honest validators among them, called \emph{honest validator set} ($HVS$). But nodes in $HVS$ may become malicious at another round. For any block $blk'>B'$, at most $f$ nodes in $HVS$ vote for $blk$ since at most $f$ nodes change,  and at least $f+k$ cannot vote for it. Hence, it's impossible to create a valid vote certificate  for $blk'$ which can unlock $u$, because at most $n-(f+k)=3f<n-f$ validators may vote for $blk'$. According Lemma \ref{lemma:lock}, $PLock(B,l-2)$ will never be unlocked any more as well.
\end{proof}

Each validator has to participant in the consensus for previous height, hence it needs to keep a propose-lock and vote-lock for previous height.  According to Lemma \ref{lemma:lock}, it is enough for each validator to keep a vote-lock for last height and a propose-lock for current height. For validator $u$ who runs in round  ${R}_{l,v}$,  the block that $u$ is vote-locked on at height $l'$ ($l'\le l$) is denoted by $VL_{l'}^{(u)}$. According to Lemma \ref{lemma:unlock}, block $VL_{l'}^{(u)}$ ($l'\le l-2$) will not be changed any more.

\begin{lemma}\label{lemma:parallel}
	If honest validator $u$ runs at round ${R}_{l,v}$($l>2$) and another honest validator $u'$ runs at round ${R}_{l',v'}$ such that $l'\ge l$,  $VL_{l-2}^{(u)}$ is the block that $u$ locked on, then $u'$ must be locked on $VL_{l-2}^{(u)}$ at height $l-2$. 
\end{lemma}

\begin{proof}
	By the way of contradiction, assume validator $u'$ is locked on another block $VL_{l-2}^{(u')}$ such that $VL_{l-2}^{(u)} \vdash VL_{l-2}^{(u')}$. There must be a height number $h$($h\le l-2$) satisfies that the ancestors of block $VL_{l-2}^{(u)}$ and $VL_{l-2}^{(u')}$ have identical \emph{PreHash} at height $h$. Therefore, block $VL_{h}^{(u)}$ and $VL_{h}^{(u')}$ are comparable. It may assume that $VL_{h}^{(u)} < VL_{h}^{(u')}$, according to the proof Lemma \ref{lemma:unlock}, the vote certificate for $VL_{h}^{(u')}$ cannot be derived by any validator. Consequently, $u'$ cannot move to height $l'$. Hence, the origin lemma is held.
\end{proof}

Lemma \ref{lemma:parallel} indicates that any two honest validators must vote-lock on same block which is finalized after two block heights. This property is vital to the correctness of consensus protocol. 


\subsection{Handling Proposals and Votes from Future Rounds} \label{subsec:future_vote}
If a validator restarts from crash, it may receive proposals or votes with valid vote certificate from future round. When validator $u$ receives a vote certificate $Cert$ for block $B$ from a validator at height $l'$ such that $l'\ge l-1$,  $Cert$ is valid only if $B\nvdash VL_{l-2}^{(u)}$ according to Lemma \ref{lemma:parallel}.  At the beginning, we define the relationship between validator and vote certificate in Definition \ref{def:succeed}.
\begin{definition} \label{def:succeed}
	We say a vote certificate $Cert$ for $B$ at round ${R}_{l',v'}$ succeeds  validator $u$ running at ${R}_{l,v}$ represented as $Cert \succ u$, if one of the following conditions hold (\romannumeral1)  $Cert \nvdash VL_{l-1}^{(u)}$, $l'>l$; (\romannumeral2) $Cert \nvdash VL_{l-1}^{(u)}$, $l'=l$ and $u$ hasn't vote for any block $B'$ that $B'>B$; (\romannumeral3) $Cert \vdash VL_{l-1}^{(u)}$, the ancestor of $B$ at height $l-1$ is $blk$, then $blk>VL_{l-1}^{(u)}$.
\end{definition} 

 If a validator $u$ receives a proposal or vote with vote certificate succeeds itself, it means that $u$ has fallen behind others. As shown in Figure \ref{fig:partition}, validator $u_{4}$ is the collector for round ${R}_{10, 0}$, but it is running at height 8.  Validator $u_4$ receives a vote from $u_2$ with a $Cert$ succeeding itself. If $u$ receives a proposal containing $Cert_{l',v'}$ succeeding itself, then it enter round ${R}_{l'+1,0}$ after synchronization from others. Similarly, validator synchronizes data from others upon receiving vote $V_{l',v'}$ with $Cert$ succeeding itself. If validator is the collector for future round ${R}_{l',v'}$, it broadcasts the certificate to help others catch up only once. However, fault validators can create votes for future round easily and let honest validators broadcast votes. For example, $f$ honest validators run at round ${R}_{l,v}$, but the others run at round ${R}_{l+2, 0}$. Then faulty validators can send votes to $C_{l+2, 1}, C_{l+2, 2}, \cdots$ at the same time, and if they run at height $l$, they broadcast vote certificate after synchronization. In this situation, the complexity of communication is $O(nf)=O(\frac{1}{3}n^2)$. We propose a mechanism to avoid $O(n^2)$ transmissions. 
\begin{subequations} \label{equation:timeout}
	\setlength{\abovedisplayskip}{3pt}
	\setlength{\belowdisplayskip}{3pt}
	\begin{align} 
	&TO_{c}(i)=2^{i}\Delta \\
	&TO_{p}(i)=TO_{c}(i)+2^{i}\Delta=2^{i}\cdot 2\Delta
	\end{align}
\end{subequations}

It is known that each collector is elected by \emph{VRF} according to Equation (\ref{equation:collector}). For round ${R}_{l,v}$, there is a  unique collector $C_{l,v}$. If adversary wants to let honest validator $u$ broadcast vote certificate, it must calculate a  correct \emph{VRF} for a round whose collector is $u$, otherwise $u$ cannot make sure it's the collector for that round. To avoid starting a view change too soon, the \emph{Timeout} for the next round doubles if timer for current round expires in LinSBFT. We define the \emph{Timeout} $TO_{c}(i)$ and $TO_{p}(i)$ for the $i-$th round at each height in Equation \ref{equation:timeout}. The collector waits $TO_{c}(i)$ time and enter the $(i+1)$-th round without collecting $n-f$ votes. Validators enter $(i+1)$-th after $TO_{p}(i)$ time without receiving proposal from collector. For round ${R}_{l',v'}$, adversary needs to wait at least  $\sum_{i=0}^{v'-1}{TO_{p}(i)}$ before sending vote $V_{l',v'}$ to $C_{l',v'}$. When a validator receives vote with $Cert'$ from future round ${R}_{l',v'}$,  it verifies whether it is the collector for that round at first. According to Lemma \ref{lemma:unlock}, $Cert'$ cannot conflict with block $VL_{l-2}^{(u)}$, otherwise it is invalid. Then it calculates the minimized transmit time $DUR(l-2,l',v')$ from the creation of block $VL_{l-2}^{(u)}$ to round ${R}_{l',v'}$ as Equation (\ref{equation:duration}). 
\begin{equation} \label{equation:duration}
\small
\begin{aligned} 
DUR(h, h',r')&=(h'-h)\Delta t_{f}+\sum\limits_{i=0}^{r'-1}{TO_{p}(i)}\\
&=(h'-h)\Delta t_{f}+\sum\limits_{i=0}^{r'-1}{2^{i}\cdot 2\Delta}\\
&=(h'-h)\Delta t_{f}+\left(2^{r'}-1\right)\cdot 2\Delta
\end{aligned}
\end{equation}

In Equation (\ref{equation:duration}), $\Delta t_{f}$ is the minimized time to finish a round to achieve consensus. The duration is at least $(h'-h)t_{f}$ from round ${R}_{h,0}$ to round ${R}_{h',0}$ and the minimized transmit time from round ${R}_{h',0}$ to round ${R}_{h',r'}$ is  $\sum_{i=0}^{v'-1}{TO_{p}(i)}$. For validator $u$, the duration from the time of creation of $VL_{l-2}^{(u)}$  to now is $\Delta T$. If $\Delta T<DUR(l-2,l',v')$, validator just rejects the vote and do nothing since nobody can move to round ${R}_{l',v'}$ within $\Delta T$. As presented in Lemma \ref{lemma:futurevote}, the complexity of transmissions caused by future votes is $O(n)$. The proof of Lemma \ref{lemma:futurevote} is shown in Appendix \ref{append:future}.

\begin{lemma} \label{lemma:futurevote}
	With the negligible probability, the number of validators who broadcast vote certificate is constant.
\end{lemma}

%% file: protocol.tex
\section{Complete Protocol} \label{sec:protocol}

\subsection{Consensus Algorithm} \label{subsec:protocol}
\begin{algorithm}[t] 
	\footnotesize
	\caption{LinSBFT protocol framework}
	\textbf{local variables} \\
	${R}_{h,r}$: the round $u$ runs at \\
	$r'$: the last round at height $h-1$\\
	${T_{c}, T_{p}}$: timer for collection and proposal \\
	$B$: received block for current height of $u$\\

	\textbf{/*Event on validator*/}\\
	\Upon{{\rm reception of} $m=\left<{{PROPOSE}},P_{l,v}\right>_{s}$ } {
		$cert\leftarrow$ vote certificate in $P_{l,v}$ \\
		\If{$cert\succ u$}{
			$handleFuturePropsal(P_{l,v})$ $/* Algorithm\,2*/$\\
		} \ElseIf{$cert \nvdash  VL_{h-1}^{(u)}$} {
				$handleProposal(P_{l,v})$ $/* Algorithm\,2*/$\\
		}
	}
	\vspace{3pt}
	\Upon{{\rm timer} $T_{p}$ {\rm for} ${R}_{l,v}$ {\rm expires}} {
		\If{$l=h-1$}{
			$r'\leftarrow r'+1$\\
			$enterPrepare\left(h,r',VL_{h-1}^{(u)}\right)$ $/* Algorithm\,3*/$\\
		} \ElseIf{$l=h$} {
			$r\leftarrow r+1$\\
			$enterPrepare(h,r,B)$ $/* Algorithm\,3*/$\\
		}
	}
	\vspace{3pt}
	\Upon{{\rm reception of} $m=\left<CERT, Cert_{l,v}\right>_{s}$ } {
		\If{$Cert_{l,v}\succ u$}{
			$synchronize(Cert_{l,v},l,v)$ $/* Algorithm\,3*/$\\
		}
		\If{$Cert_{l,v}>PL_{h}^{(u)}$}{
			$PUnlock(h)$, $PLock(B_{l,v},h)$\\
		}
	}
		
	\vspace{5pt}
	\textbf{/*Event on collector(leader)*/} \\
	\Upon{{\rm reception of} $m=\left<VOTE,V_{l,v}\right>_{s}$ }{
		$cert\leftarrow$ vote certificate in $V_{l,v}$ \\
		\If{$cert\succ u$}{
			$handleFutureVote(V_{l,v}, cert)$ $/* Algorithm\,2*/$\\
		} \ElseIf{$cert \nvdash  VL_{h-1}^{(u)}$} {
				$handleVote(V_{l,v})$ $/* Algorithm\,2*/$\\
		}
	} 
	\vspace{3pt}
	\Upon{{\rm timer} $T_{c}$ {\rm for} ${R}_{l,v}$ {\rm expires}} {
		$enterPropose(l,v+1)$ $/* Algorithm\,3*/$\\
	}
\end{algorithm}
\begin{algorithm}[t] 
	\footnotesize
	\setlength{\belowdisplayskip}{3pt}
	\caption{Handle for proposal and vote}
	(local variables are the same as Algorithm 1)\\
	\textbf{/*Procedure on validator*/}\\
	\Procedure{handleProposal$(P_{l,v})$}{
		\If{${R}_{h+1,0}={R}_{l,v}$}{
			$VLock(B, h)$, $B\leftarrow$ block of $P_{l,v}$, $PLock(B, h+1)$\\
		    $h\leftarrow l, r'\leftarrow r, r\leftarrow v$  \\
			\textbf{/*finalize block after two phases of voting*/} \\
			$enterFinalize\left(VL_{l-2}^{(u)}\right)$ $/* Algorithm\,3*/$\\
			$enterPrepare(h,r,B)$ $/* Algorithm\,3*/$\\
		} \ElseIf{${R}_{h,r+1}={R}_{l,v}$}{
		\textbf{/*$P_{l,v}$ cannot unlock validator*/} \\
			$r\leftarrow r+1$, $B\leftarrow$ block of $P_{l,v}$\\
			$enterPrepare(h,r,B)$ $/* Algorithm\,3*/$\\
		} \ElseIf{${R}_{h-1,r'+1}={R}_{l,v}$}{
			$r'\leftarrow r'+1$\\
			send vote for $VL_{h-1}^{(u)}$ to collector $C_{h-1, r'}$
		}
	}
	\vspace{3pt}
	\Procedure{handleFutureProposal$(P_{l,v}, cert)$} {
		$synchronize(cert, l, v)$ $/* Algorithm\,3*/$\\
		$handleProposal(P_{l,v})$ $/* Algorithm\,3*/$
	}
	
	\vspace{5pt}
	\textbf{/*Procedure on collector(leader)*/} \\
	\Procedure{handleVote$(V_{l,v})$} {
		add $V_{l,v}$ to vote set for ${R}_{l,v}$ \\
		\If{{\rm collect} $n-f$ {\rm votes for} ${R}_{l,v}$}{
			$enterPropose(l+1,0)$ $/* Algorithm\,3*/$\\
		} 
	}
	\Procedure{handleFutureVote$(V_{l,v}, cert)$} {
	    \textbf{/*$u$ has fallen behind others*/}\\
		\If{$DUR(h-2,l,v)$ {\rm is valid}}{
			$synchronize(cert, l, v)$ $/* Algorithm\,3*/$\\
			broadcast message $\left<CERT, cert\right>$\\
			$handleVote(V_{l,v})$ $/* Algorithm\,2*/$
		}
	}
\end{algorithm}
\setlength{\textfloatsep}{0.1cm} 
\setlength{\floatsep}{0.1cm}

In addition to steps of \emph{Propose} and \emph{Prepare}, we introduce a new step called \emph{Finalize} in which a validator finalizes block and executes transactions of it, resulting in the transition of local state machine. In general, for each round, every validator experiences threes steps in order: \emph{Propose}, \emph{Finalize} and \emph{Prepare}. In \emph{Propose} step, collector(leader) proposes block, others receive them from collectors. In \emph{Finalize} step, all validators finalize block proposed two height before. After that, in \emph{Prepare} step, every validator sends vote to collector and collector collects these votes to derive a threshold signature. 

Algorithm 1 describes the framework of LinSBFT protocol and Algorithm 2 is the details of processes for proposal and vote used in Algorithm 1. Algorithm 3 presents necessary procedures used in Algorithm 1 and Algorithm 2, including the  transfer functions among three steps and the synchronize procedure. These events and procedures are divided into two kinds, one runs on normal validators (including the collector) and the other runs on collector(leader). For brevity, the verification and check of messages are omitted. As discussed before, each validator participants the protocol for previous height to help others move to larger height. Therefore, each validator record the last round number $r'$ at previous height. In LinSBFT, any process is triggered by some given events such as reception of message and expiry of timer. We discuss the process for reception of \emph{Proposal}, \emph{Vote} and \emph{Cert} message. If timer $T_{c}$ expires, it means that collector cannot receive $n-f$ votes within a timespan, and if proposal does not arrive in time, the timer $T_{p}$ for proposal expires as well. 

At the beginning, we discuss the process of proposal broadcast by leader of each round.  First, when validator $u$ receives a proposal $P_{l,v}$ containing a vote certificate $cert$ succeeding itself, which means $u$ has fallen behind others, it synchronizes data from others at once and enter round ${R}_{l,0}$ (Algorithm 1, line 8-10). Otherwise, if the certificate $cert$ does not conflict with locked block $VL_{h-1}^{(u)}$ of $u$, it handles the proposal to transmit its state machine by function $handleProposal$ (Algorithm 1, line 11-12), which is detailed in Algorithm 2. Second, in Algorithm 2, when receives  proposal containing the certificate for current round ${R}_{h,r}$(Algorithm 2, line 4-9), which is the ordinary case of LinSBFT, $u$ enters round ${R}_{h+1,0}$ and update its local variables, resulting in that $u$ is vote-locked on current block $B$ at height $h$ and propose-locked on block of $P_{l,v}$ at height $h+1$(Algorithm 2, line 4-5). Then, validator moves to \emph{Finalize} step (Algorithm 2, line 8), in which validator finalizes block $VL_{l-2}^{(u)}$ if it is not finalized before(Algorithm 3, line 13-15). In \emph{Prepare} step(Algorithm 2, line 9), $u$ sends its vote to the collector and sets a timer $T_{p}$ for proposal(Algorithm 3, 17-19). If $u$ is the collector responsible for collection of votes from others, it  sets a timer $T_{c}$ for collection(Algorithm 3, line 20-21). Third, if the received proposal is for round ${R}_{h,r+1}$, it means the collector $C_{h,r}$ cannot collect at least $n-f$ votes to derive a threshold signature before the expiry of timer  $T_{c}$, therefore $u$ enter round ${R}_{h,r+1}$ and sends votes for proposed block to collector(Algorithm 2, line 8-11). Last, when receives proposal for previous height, if the proposed block is identical to validator's vote-locked block, it sends vote   for $VL_{h-1}^{(u)}$ along with latest vote certificate  to collector $C_{h-1,r'}$ at height $h-1$(Algorithm 2, line 12-14).  Otherwise, $u$ just ignores the proposal for the safety of protocol.

\begin{algorithm}[t] 
	\footnotesize
	\setlength{\belowdisplayskip}{3pt}
	\caption{Functions for LinSBFT protocol}
	(local variables are the same as Algorithm 1)\\
	\textbf{/*Procedure on validator*/}\\
%
	
	\Procedure{enterFinalize$(blk)$} {
		\If{ $blk$ not finalized}{
			$applyBlock(blk)$\\
		}
	}
	\Procedure{enterPrepare$(l,v,B)$} {
		$vote\leftarrow createVote(B)$ \\
		send vote $vote$ to collector $C_{l,v}$\\
		$setTimer(T_{p},l,v)$ \\
		\If{$u$ \rm{is the collector for} ${R}(l,v)$}{
			$setTimer(T_{c}, l, v)$\\
		}
	}
	\Procedure{synchronize$(cert, l, v)$} {
		\If{$cert\nvdash VL_{h-1}^{(u)}$}{
			synchronize block $B_{h-1,r'}$ and $B_{h,r}$ \\
			$VUnlock(h-1)$, $VLock(B_{h-1,r'},h-1)$, $PUnlock(l)$
		}
		\For{$i$ {\rm \textbf{from}} $h+1$ {\rm\textbf{to} } $l$ }{
			synchronize block $B_{i, v_{i}}$ \\
			$VLock(B_{i-1,v_{i-1}}, i-1)$, $PLock(B_{i,v_{i}},i)$\\
			$r'\leftarrow r, r\leftarrow 0, h\leftarrow h+1$
		}
	}
	
	\vspace{5pt}
	\textbf{/*Procedure on collector(leader)*/} \\
	\Procedure{enterProprose$(l, v)$} {
		\If{$u$ {\rm not have proposed for} ${R}_{l,v}$}{
			\If{$PL_{l}^{(u)} \ne nil$ \rm\textbf{and} ${R}_{h,r+1}={R}_{l,v}$} {
				broadcast \emph{Proposal} containing $PL_{l}^{(u)}$
			} \Else{
				$cert\leftarrow createCert(l,v)$, 	$blk\leftarrow createBlock()$ \\
				$proposal\leftarrow makeProposal(blk, cert)$ \\
				broadcast $proposal$\\
			}	
		}
	}
\end{algorithm}

We then present the process of vote messages sent by validators. A validator just handles a vote message if it is the collector for round of the vote, otherwise the vote message is discarded by it. In Algorithm 1,  if $V_{l,v}$ is a vote for future round (Algorithm 1, line 29-30), $u$ changes to synchronize data from other validators after verifying the duration time(Algorithm 2, line 26-31). Upon receiving vote $V_{l,v}$ for current round, collector $u$ adds $V_{l,v}$ to local vote set.  If $u$ collects more than $n-f$ votes, it enters \emph{Propose} step for round ${R}_{h+1,0}$ and doesn't handle votes for ${R}_{h,r}$ any more(Algorithm 2, line 19-21). In \emph{Propose} step, the collector propose a new block along with a $Cert$ derived based on votes to all validators including itself(Algorithm 3, line 6-11). 

Last, we discuss the process of expiry of timer and \emph{Cert} message. If timer $T_{p}$ expires, validator just enters round ${R}_{l,v+1}$ and sends vote along with latest $Cert$ to the next collector(Algorithm 1, line 14-19). A validator proposes its vote-locked block for previous height when it becomes the collector for previous height, since it has to participate the protocol for previous height(Algorithm 1, line 14-16). If timer $T_{c}$ expires, it means as a collector, $u$ cannot receive votes from super-majority within given timespan(Algorithm 1, line 33-34). Therefore,  $u$ enters \emph{Propose} step for next round and proposes the block $PL_{l}^{(u)}$ if any(Algorithm 3, line 24-25). It notes that $u$ just creates a new block if it doesn't have a propose-lock at height $l$(Algorithm 3, 25-30). Upon receiving \emph{Cert} message with vote certificate $cert$, if $cert\succ u$, validator $u$ synchronizes data as well and update local \emph{Propose-Lock} and \emph{Vote-Lock}(Algorithm 1, line 20-24). 


\subsection{Handling Changes in Participant Set} \label{subsec:setup_ts}

\textbf{Participant set update}. 
Since there is a large setup cost for generating and exchanging keys as described in Section \ref{subsec:preliminaries}, the participant set cannot change too frequently. For instance, exchanging public keys and IP addresses between each pair of nodes already take $O(n^2)$ transmissions. In the literature, such costs are often hidden by assuming the existence of a centralized public-key infrastructure (PKI), which is not practical in a permissionless, public blockchain setting. To amortize the setup costs, LinBFT divides time into epochs of length $O(n)$, and nodes can only join or leave at the beginning of each epoch. Each update to the participant set (i.e., a validator join or leave request) is simply treated as a transaction, which will be included in a new block to be added to the blockchain. The rationale is that since LinSBFT guarantees deterministic safety and liveness, there must be deterministic consensus over the next participant set.
Specifically, at the end of epoch $\mathcal{E}$, the set of join/leave requests contained in finalized blocks determine the changes in the participant set in epoch $\mathcal{E}+1$. In the worst case, these involve $O(n)$ join/leave transactions, e.g., when the entire participant set is replaced. LinSBFT needs to run a DKG protocol to generate new public/private key pairs for each participant in epoch $\mathcal{E}+1$, which are required for creating threshold signatures. Note that LinSBFT does not update keys incrementally, i.e., the key pair of a staying participant from epoch $\mathcal{E}$ is still generated from scratch. This is because updating keys for threshold signature is tricky in general.

\textbf{Proof-of-Stake (PoS)}.
PoS is a common technique to counter Sybil attacks, i.e., one single person or entity pretends to be many participants by registering numerous accounts in the system. In the presence of such attacks, it is no longer appropriate to define BFT’s honest super-majority requirement based on the number of nodes, since multiple nodes can belong to the same entity.
PoS addresses this problem by mandating that each participant deposits to a special account a certain amount of money (call its stake) in the form of cryptocurrency tokens of the blockchain. The stake can only be withdrawn after the participant quits the protocol. Clearly, under PoS, the number of accounts that a single person or entity can register is limited by its financial resources. If a node wants to join consensus in epoch $\mathcal{E}+1$, it proposes a transaction to deposit a certain amount tokens in epoch $\mathcal{E}$. The top $n$ nodes who deposit most tokens are considered as the validators in consensus of epoch $\mathcal{E}+1$. 

PoS can also be implemented with different participants staking different amount of tokens, and having influence proportional to their respective stakes. Instead, in LinSBFT, all validators promise equal influence to consensus. In order to provide an incentive, LinSBFT protocol rewards the validators who participant in consensus of current epoch. There are two sources for this reward $W$: \emph{transaction fees} and \emph{coinbase}, which a continuous supply of new coins without mining in PoW. The incentive $W$ of block $B$ is owned by the leader who propose $B$. In some system (e.g. Ethermint \footnote{https://github.com/cosmos/ethermint}), the reward of each block is shared by all validators, where a malicious node still gains reward according to protocol although it stops working after joining in consensus in epoch $\mathcal{E}$. In LinSBFT, every round is assigned a new leader via VRF which is unpredictable in advance, thereby rewards of all validators are equal in expectation. Furthermore, if hope the reward of each validator is proportional to its deposition, we can adjust the probability that  a validator is elected as the leader/collector. Suppose the amount of token deposited by node $i$ ($0\le i \le n-1$) is a integer $d_i$, and the sum is $D=\sum_{i=0}^{n-1}d_{i}$. Then, node becomes the leader/colletor $C_{l,v}$ when $H(Cert_{l-2,v'}|v)\mod n$ is greater than $\sum_{j=0}^{i-1}d_j$ and no more than $\sum_{j=0}^{i}d_j$.  By this method, an absent validator cannot gain any reword according to incentive mechanism of LinSBFT.

%% file: correct.tex
\section{Correctness} \label{subsec:correctness}

\begin{figure*}[t]
	\setlength{\abovecaptionskip}{0.cm}
	\centering
	\subfigure[$TXs=2000$]{\label{subfigure:ordinary_tps_2000}
		\includegraphics[width=1.65in,height=1.3in]{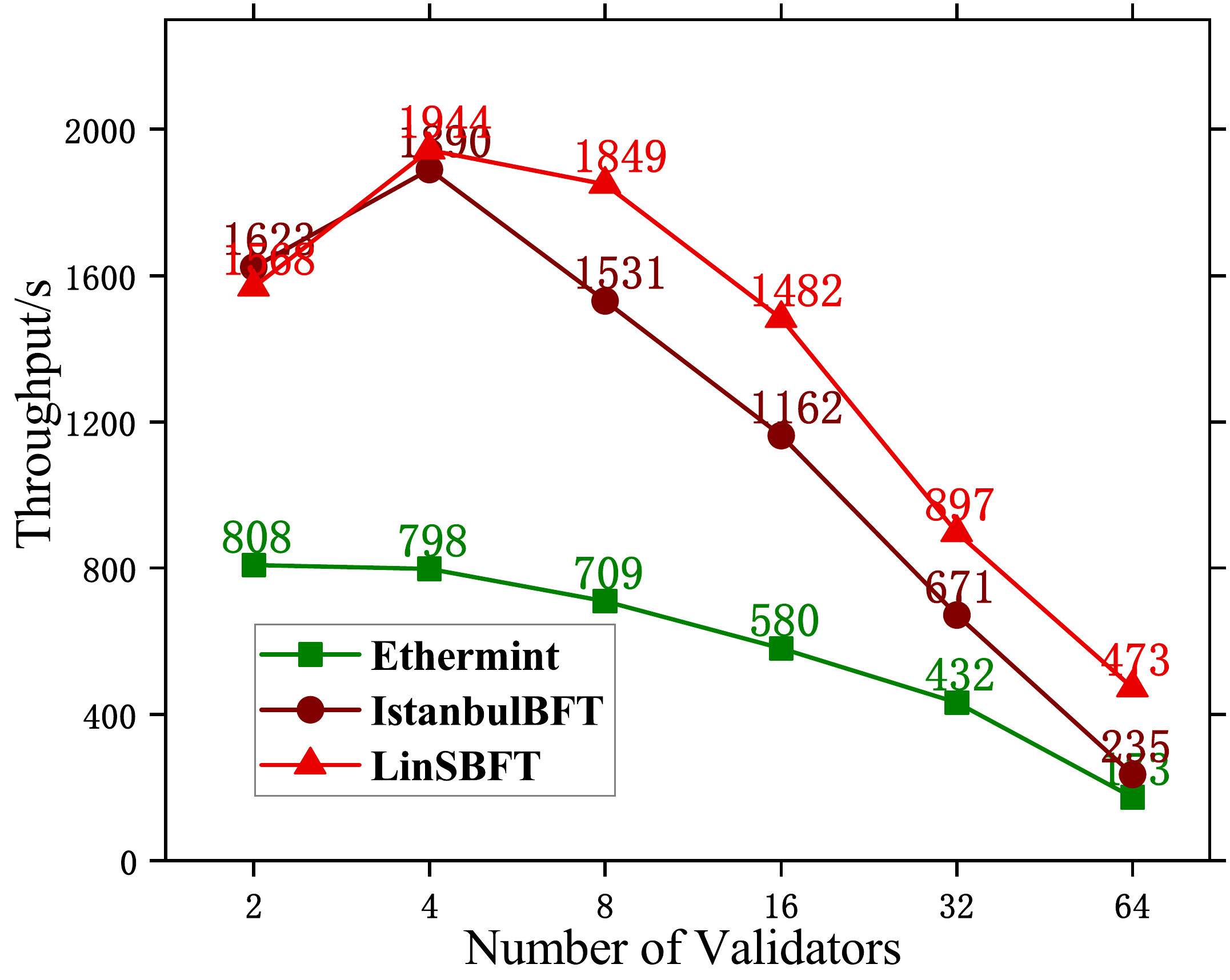}}
	\subfigure[$TXs=4000$]{\label{subfigure:ordinary_tps_4000}
		\includegraphics[width=1.65in,height=1.3in]{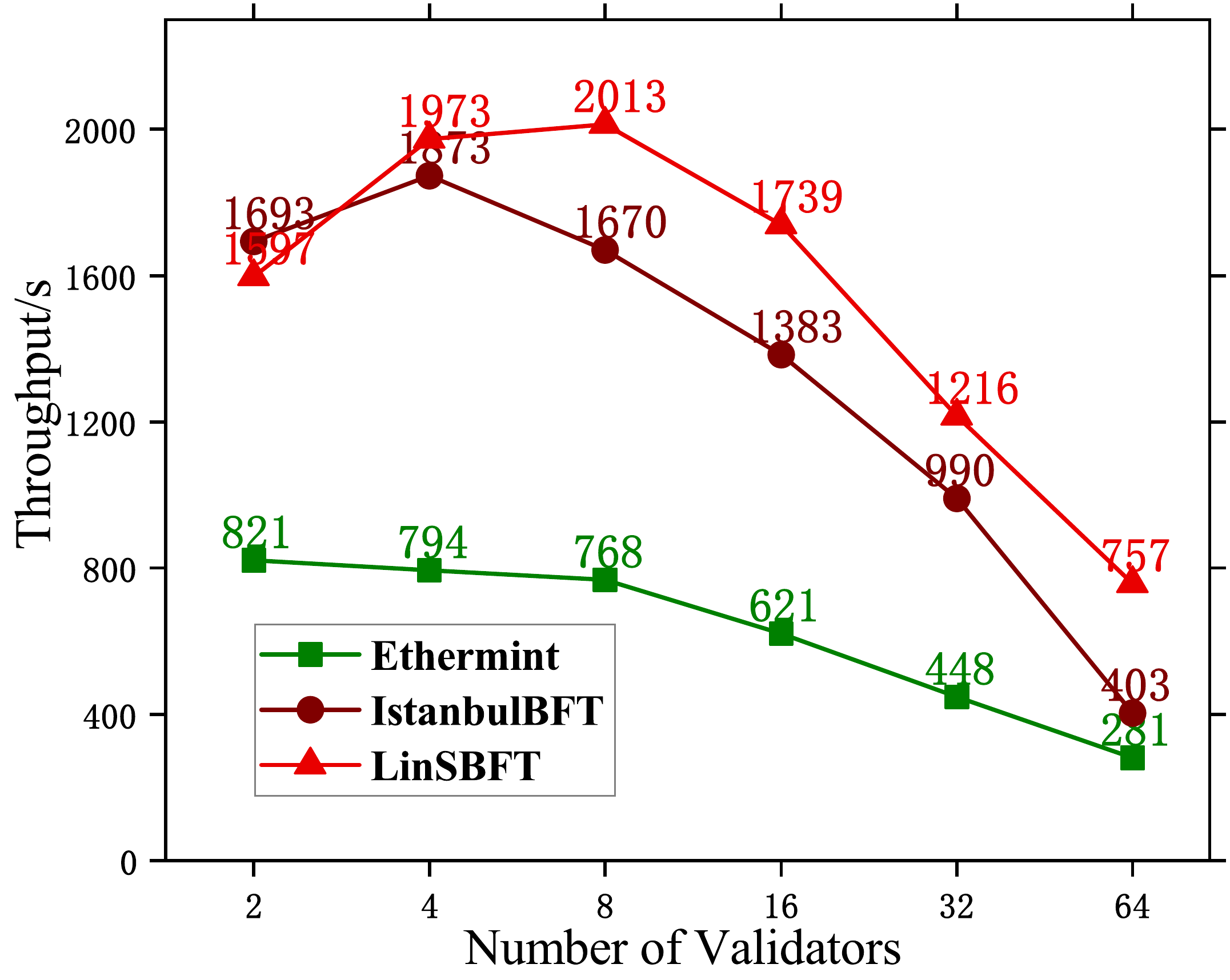}}
	\subfigure[$TXs=8000$]{\label{subfigure:ordinary_tps_8000}
		\includegraphics[width=1.65in,height=1.3in]{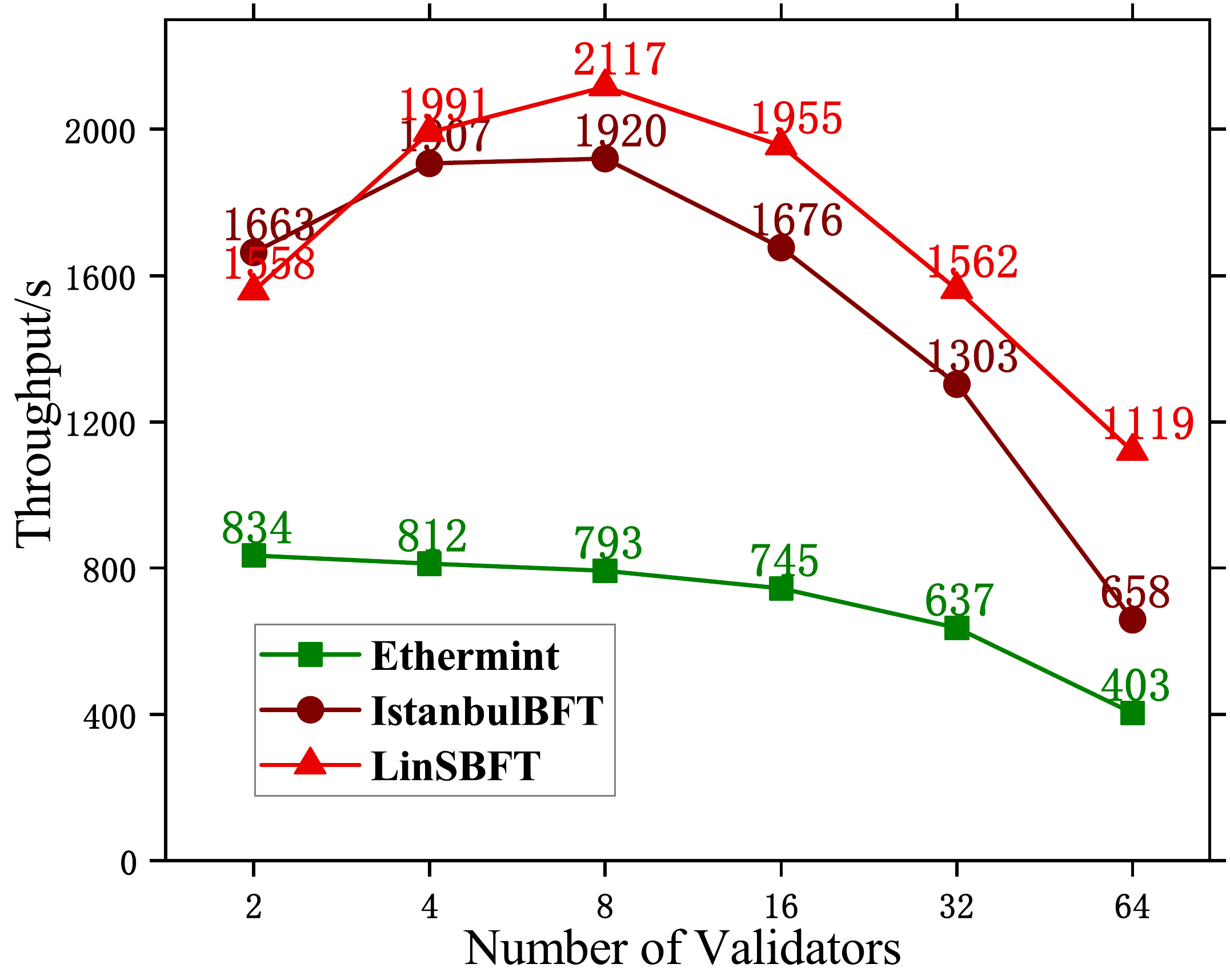}}
	\subfigure[$TXs=4000$, 1Mbps bandwidth]{ \label{subfigure:ordinary_4000_1Mbps}
		\includegraphics[width=1.65in, height=1.3in]  {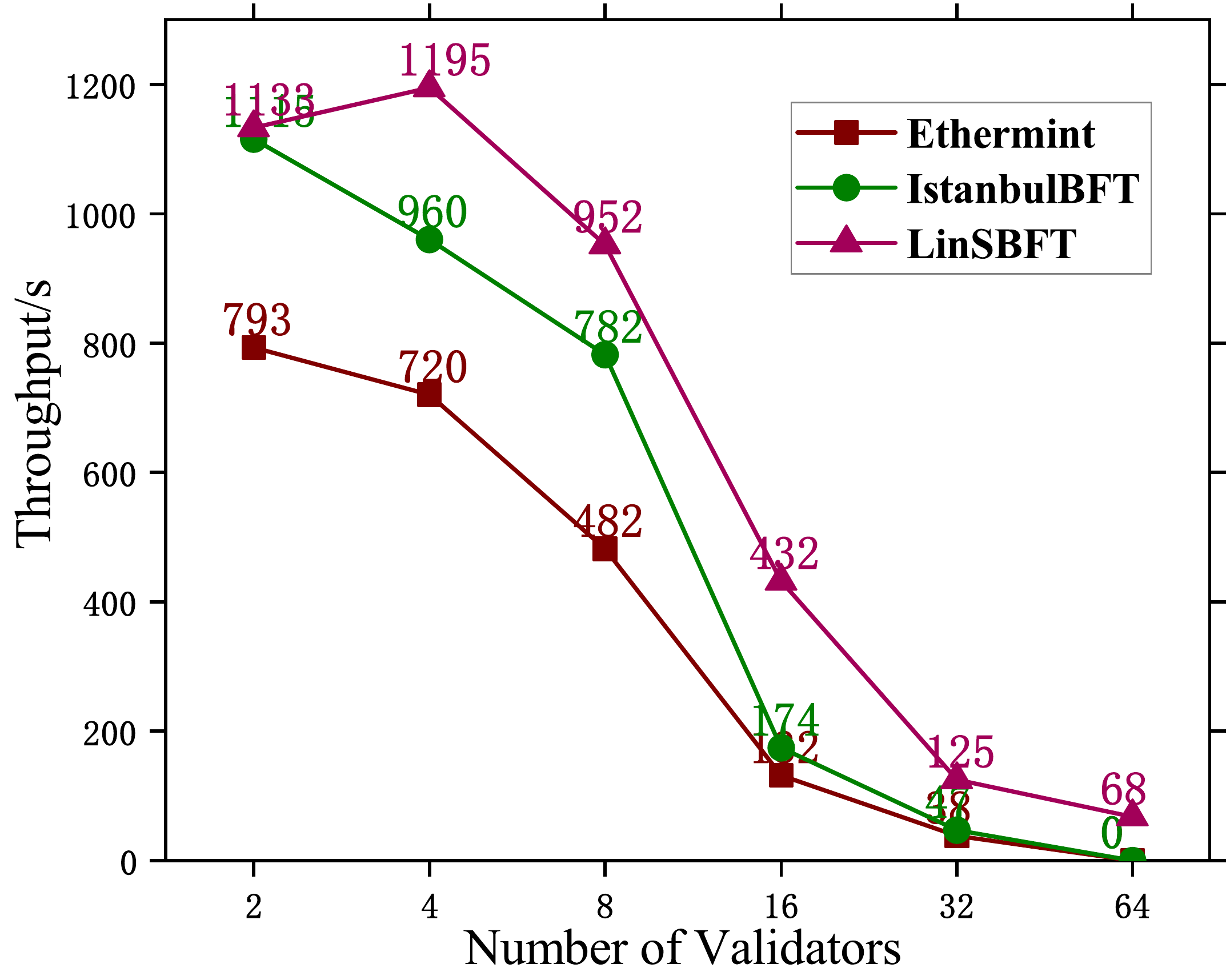}}
	\caption{\label{figure:tps}Throughput in the ordinary case with varying block size and number of validators.}
\end{figure*}

\begin{figure*}[t]
	\setlength{\abovecaptionskip}{0.cm}
	\centering
	\subfigure[$TXs=2000$]{\label{subfigure:ordinary_time_2000}
		\includegraphics[width=1.65in,height=1.3in]{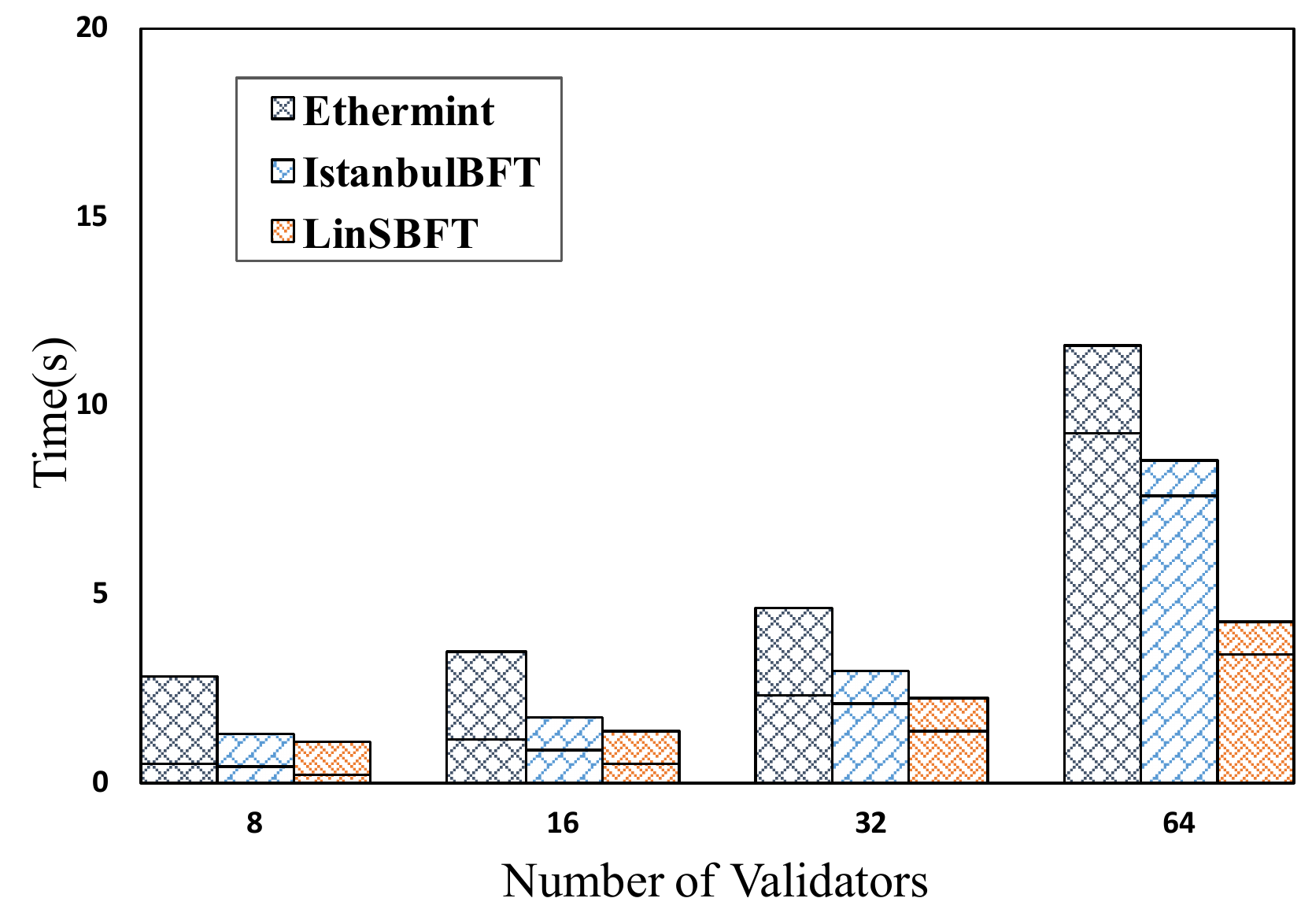}}
	\subfigure[$TXs=4000$]{\label{subfigure:ordinary_time_4000}
		\includegraphics[width=1.65in,height=1.3in]{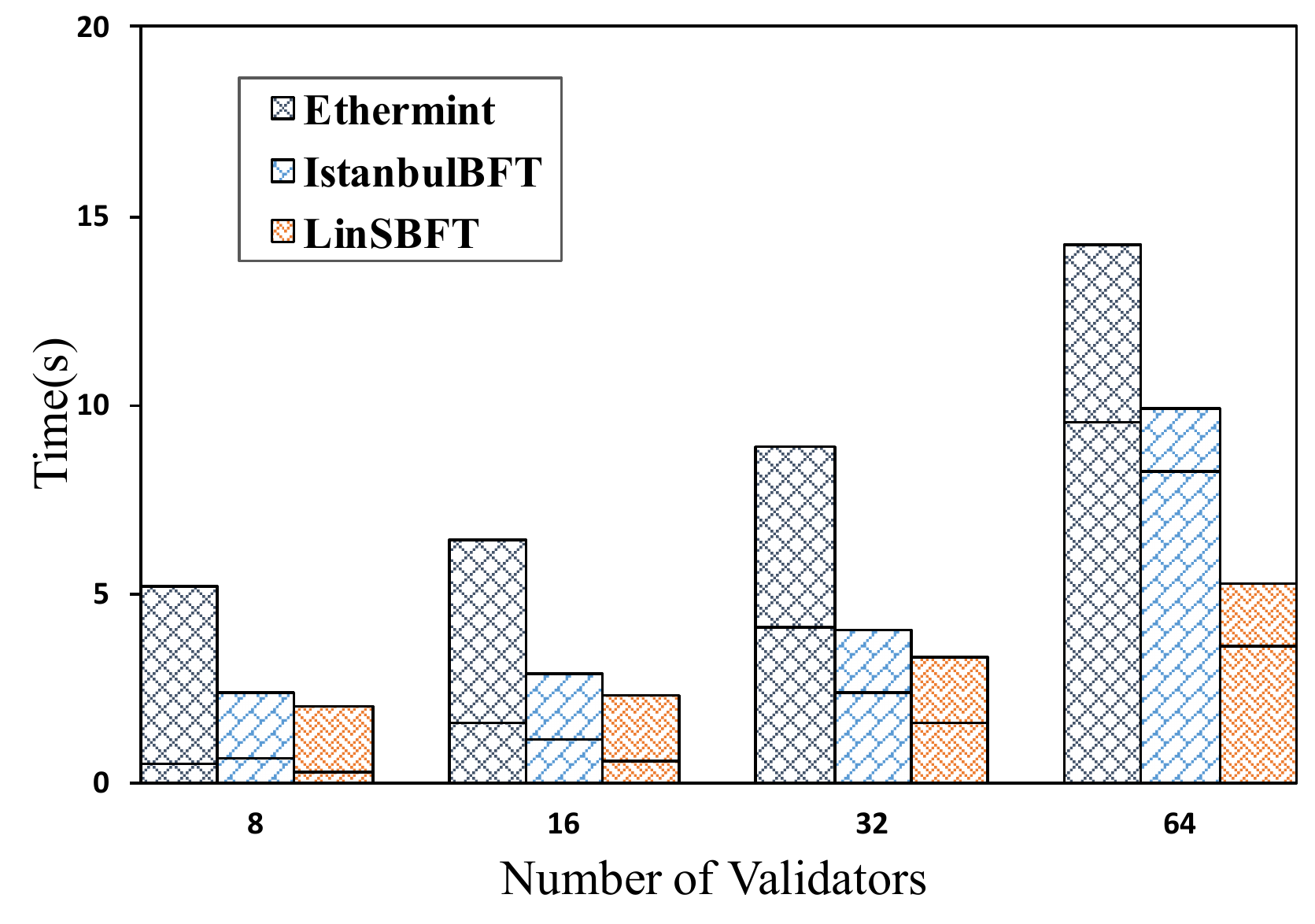}}
	\subfigure[$TXs=8000$]{\label{subfigure:ordinary_time_8000}
		\includegraphics[width=1.65in,height=1.3in]{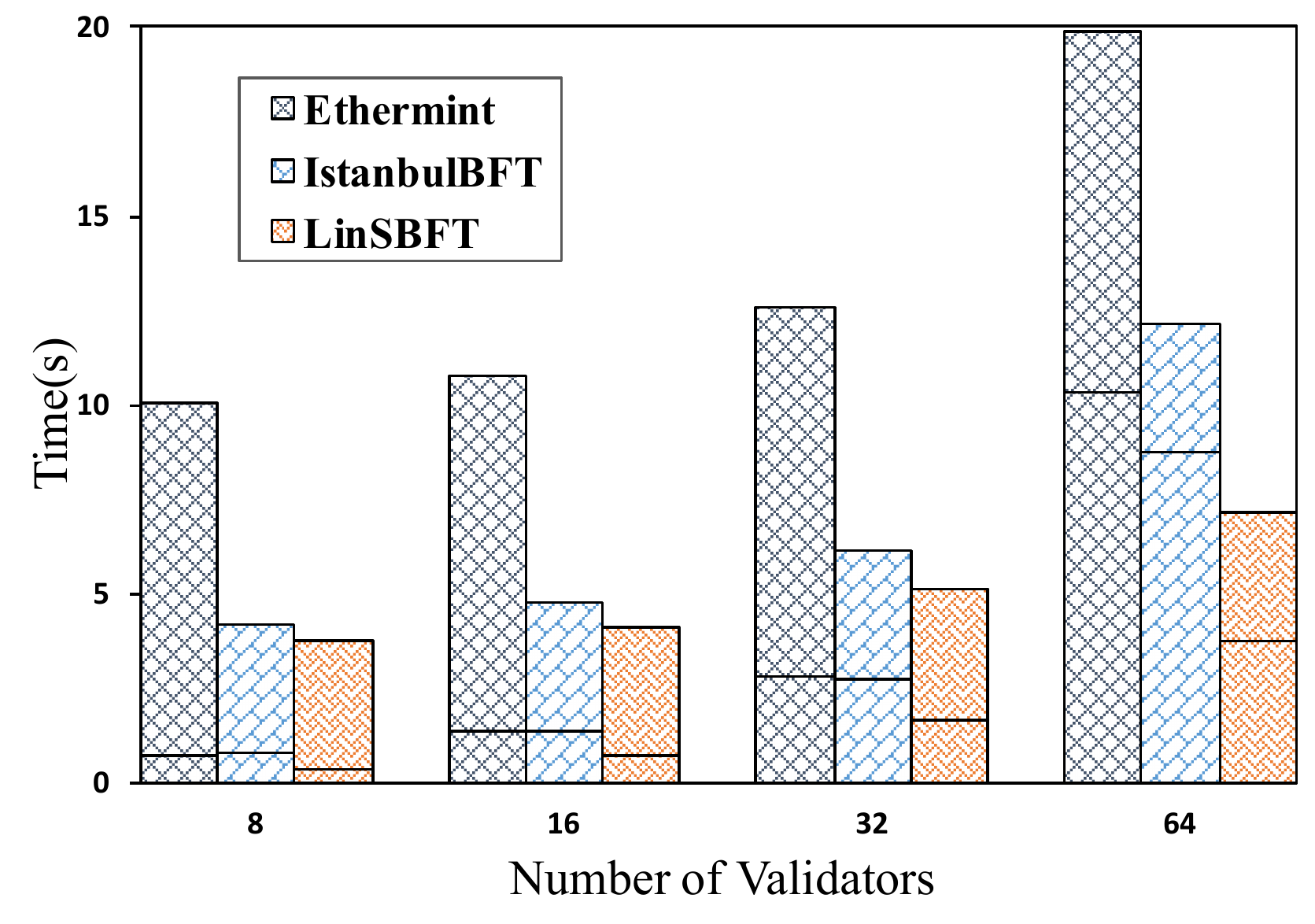}}
	\subfigure[ $TXs=4000$,1Mbps bandwidth]{ \label{subfigure:ordinary_4000_1Mbps_time}
		\includegraphics[width=1.65in, height=1.3in]  {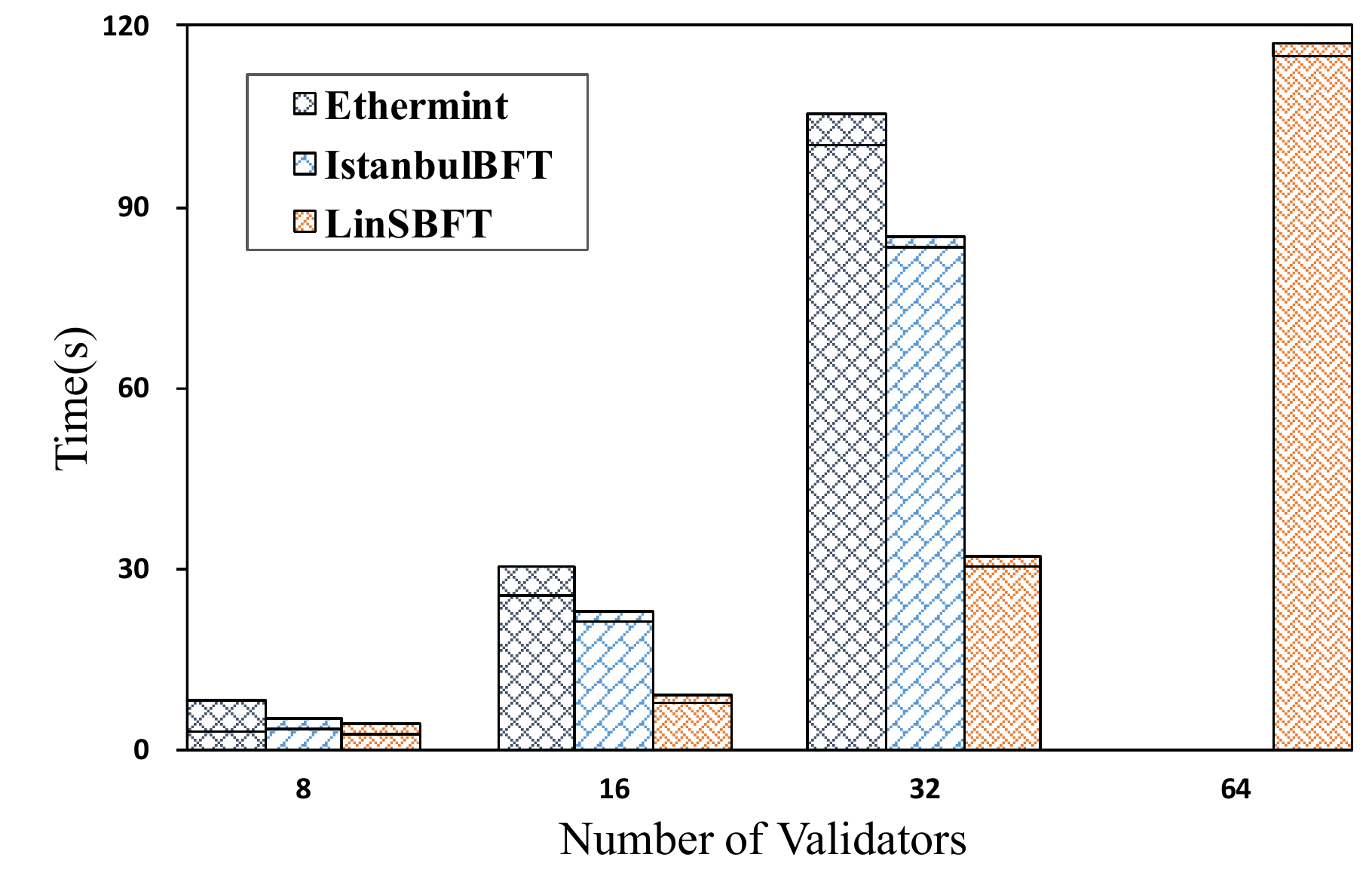}}
	\caption{\label{figure:time_cost} Latency in the ordinary case with varying block size and number of validators.}
\end{figure*}

This section sketches the proofs that LinSBFT provides safety, liveness and linear complexity.

\noindent\textbf{Safety}: Given $n\ge4f+1$, at each height, at most one  block is finalized and added to the blockchain. Claim \ref{claim:safety} establishes the safety of LinSBFT.

\begin{claim}\label{claim:safety}
	Let $B_{l,v}$ be a block finalized by an honest validator $u$. Then, there is no other block $B_{l',v'}$ that can be finalized by any honest validator such that $l'\le l $ and $B_{l,v} \vdash B_{l',v'}$.
\end{claim}
\begin{proof}
	By assumption, $u$ at least runs at height $l+2$ due to the finalization of $B_{l,v}$. Then any locked block $VL_{h}^{(u)}(h\le l)$ never be unlocked any more according to Lemma \ref{lemma:unlock}. Hence, $B_{l',v'}$ will never be finalized by $u$. For other honest validator,	by the way of contradiction, assume the block $B_{l',v'}$ is finalized by honest validator $u'$.  According to  the finalization of $B_{l',v'}$,  $u'$ must be locked on $B_{l',v'}$ at height $l'$.  The ancestor of $B_{l,v}$ at height $l'$ is $VL_{l'}^{(u)}$, then it must be that $VL_{l'}^{(u)} \vdash B_{l',v'}$ since $B_{l,v} \vdash B_{l',v'}$. But this is impossible according Lemma \ref{lemma:parallel}.
\end{proof}
Note that when a malicious validator becomes honest, it will drop all invalid blocks finalized within the period of fault. Therefore, there is single chain of blocks in the view of any validator who is honest at any round. 

\noindent\textbf{Liveness}: As discussed in Section \ref{subsec:viewchange}, to provide liveness, \emph{view change} is triggered and the protocol enters a new round when the timer $T_{p}$ expires. It is important to maximize the period of time when at least $n-f$ honest validators are in the same round, and to ensure that this period of time increases exponentially until a valid proposal is received. 

Like in PBFT, the \emph{Timeout}  $TO_{p}(v+1)$ for round ${R}_{l,v+1}$ doubles if the timer expires for round ${R}_{l,v}$ to ensure as many validators as  possible  enter same round. If a validator falls behind others, it synchronizes data from others to enter the latest round upon receiving a $Cert$ succeeding itself. The unlocking mechanism also promises the liveness that prevents a validator from being locked forever. In addition, for every round, the collector is elected by VRF. Therefore, the probability that collector is malicious for more than $I$ consecutive rounds becomes negligible (i.e., smaller than a given $\rho$) such that ($I > \log_p \rho$).

A realistic assumption made in in-production systems, e.g. Google Spanner\cite{DBLP:journals/tocs/Spanner}, is that validators have access to a globally synchronized clock with a known bounded skew. We implement the periodic synchronization based on synchronized clock to help the delayed validators catch up. For every time period $T=O(n)$, every validator broadcasts a \emph{State} message for synchronization. In Equation (\ref{equation:state}), ${R}_{l,v}$ is the round in which a validator runs at, and $Cert$ is the latest commit certificate owned by the validator. Therefore, each validator broadcasts a \emph{State} message after time $xT$($x=1,2,...$).
\begin{equation}\label{equation:state}
S=\left<Cert, l, v\right>
\end{equation}
First, whenever a validator receives a \emph{State} message with $Cert$ succeeding itself, it synchronizes data from others as before. Once a validator receives $n-f$ same $Cert$ for a block at the current height, it accepts $Cert$ and move to next height. 
Second, upon receiving $n-f$ \emph{State} messages with larger round number than itself at current height,  validator jumps to the round with smallest round number. 

The manner of periodic synchronization guarantees that the delayed validator knows the latest state of consensus and takes its initiative to synchronize data from others to enter larger height. The Claim \ref{claim:liveness} guarantees the liveness of LinSBFT, and the proof is presented in Appendix \ref{append:liveness}.
\begin{claim}\label{claim:liveness}
	In a partially synchronous network, LinSBFT reaches consensus for any block height within finite time.
\end{claim}

\noindent\textbf{Linear complexity}: LinSBFT terminates at each block height after amortized-$O(n)$ transmissions with the tricks of LVC, threshold signature and VRF, unless with negligible probability.  Claim \ref{claim:linear} states the linear complexity of LinSBFT.  Note that the messages sent by malicious validators are not counted and we only consider the messages sent during the network is synchronized. In an asynchrony network, any agreement cannot be reached, therefore it makes no sense to count these messages. 

\begin{claim}\label{claim:linear}
	In a partially synchronous network, unless with negligible probability, LinSBFT terminates after amortized-$O(n)$ transmissions at each block height.
\end{claim}

\begin{proof}
	In the ordinary case, the size of threshold signature is constant and the consensus is achieved with single round.  For malicious collector, it may take no more than constant of  rounds to reach agreement as discussed in Section \ref{subsec:viewchange}. As discussed in Section \ref{subsec:future_vote}, faulty validators can lead honest validators to broadcast proposals by sending future votes to them.  According to Lemma \ref{lemma:futurevote}, the number of validators  broadcasting messages is constant. The probability that the next collector is malicious becomes negligible after a constant number of collector changes according to the VRF. Therefore, in total, the transmissions for the consensus of each block height are $O(n)$.
\end{proof}

%% file: experiment.tex
\section{Experimental Evaluation} \label{sec:evaulation}

\begin{figure*}[t]
	\centering
	\setlength{\abovecaptionskip}{0.cm}
	\subfigure[$TXs=2000, \lambda=10s$, crash fault]{\label{subfigure:fault_2000_10s}
		\includegraphics[width=1.65in,height=1.3in]{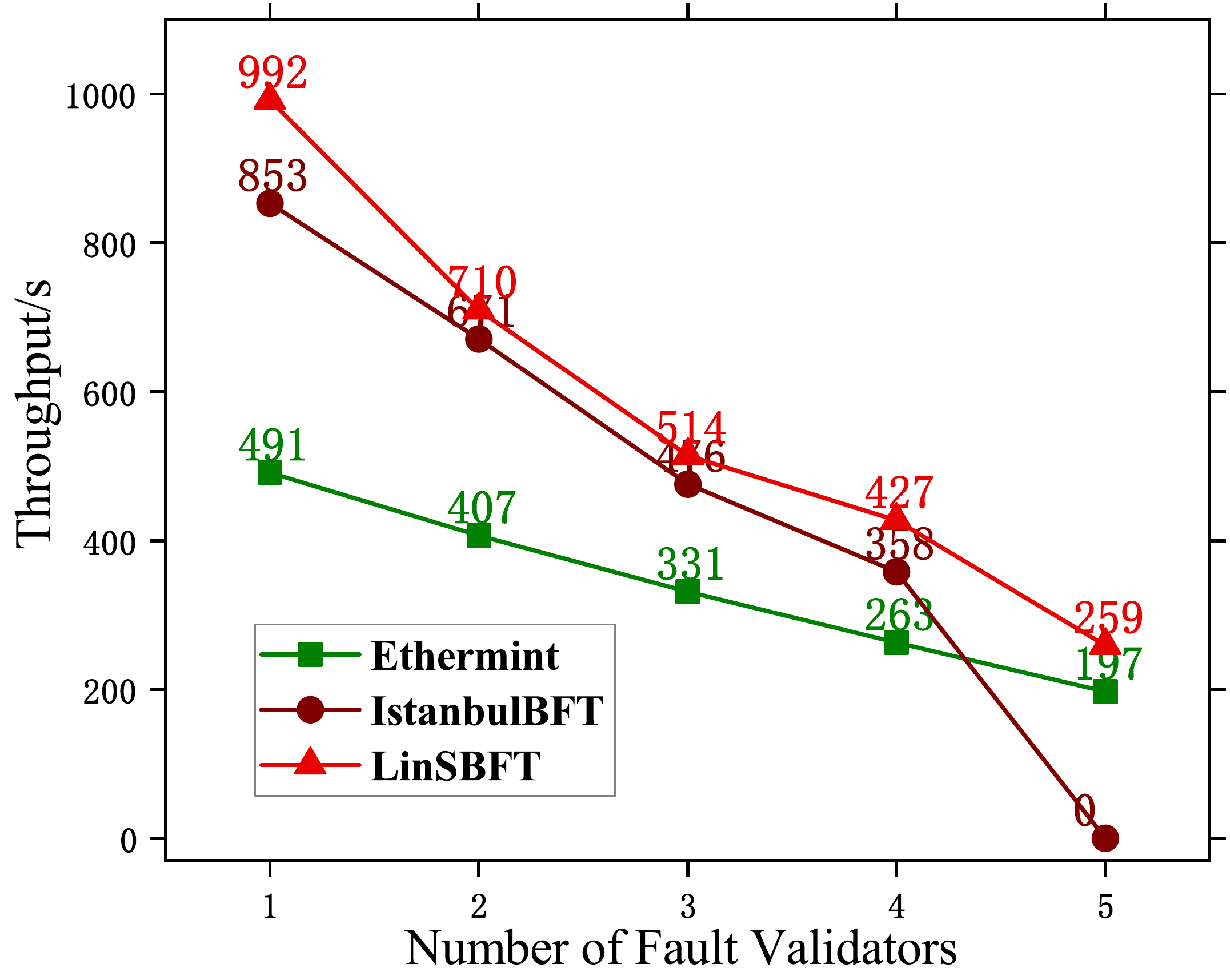}}
	\subfigure[$TXs=2000, \lambda=3s$, crash fault]{\label{subfigure:fault_2000_3s}
		\includegraphics[width=1.65in,height=1.3in]{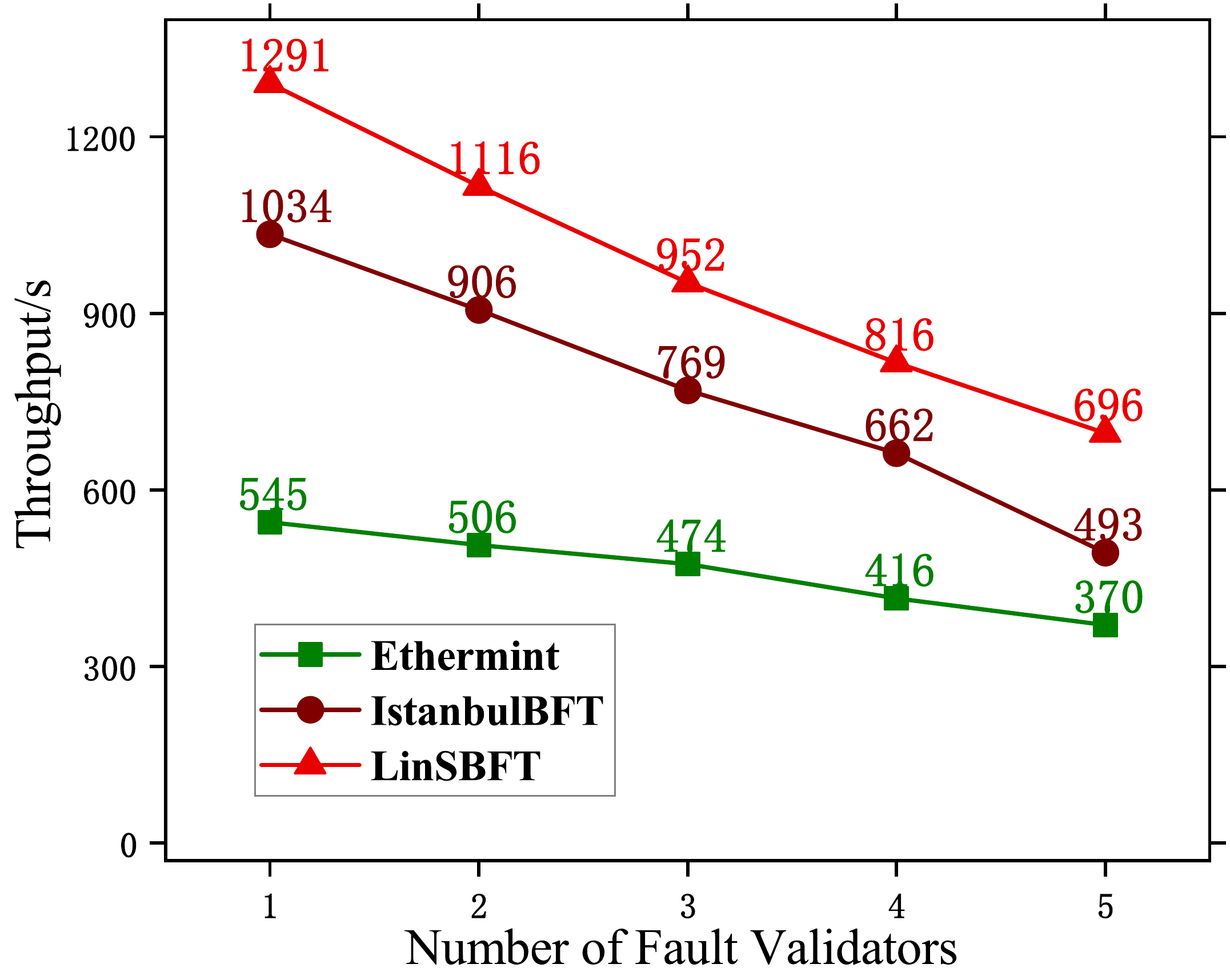}}
	\subfigure[$TXs=4000, \lambda=3s$, crash fault]{ \label{subfigure:fault_4000_3s}
		\includegraphics[width=1.65in,height=1.3in]{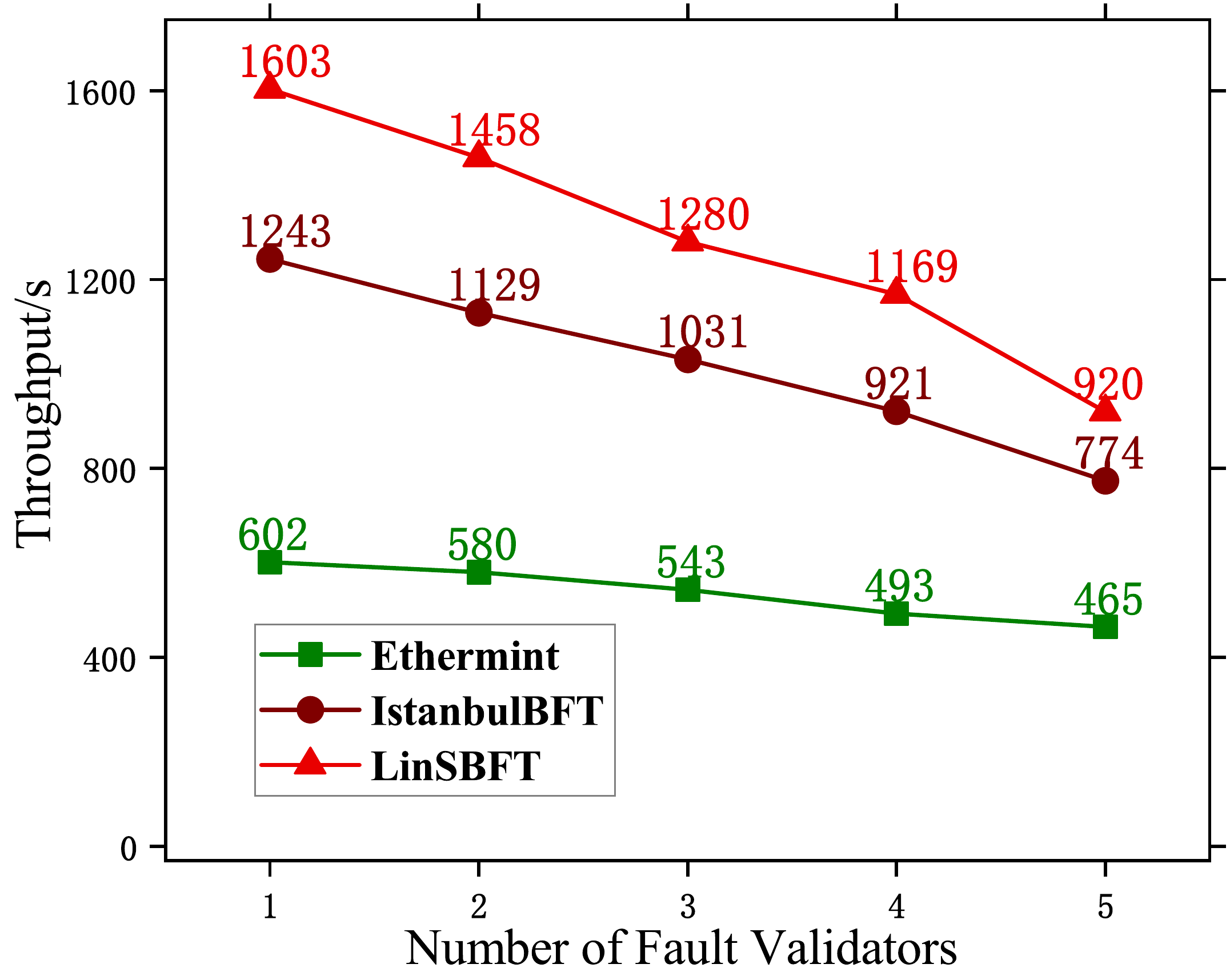}}
	\subfigure[$TXs=2000, \lambda=10s$, crash fault]{\label{subfigure:fault_2000_10s_time}
		\includegraphics[width=1.65in,height=1.3in]{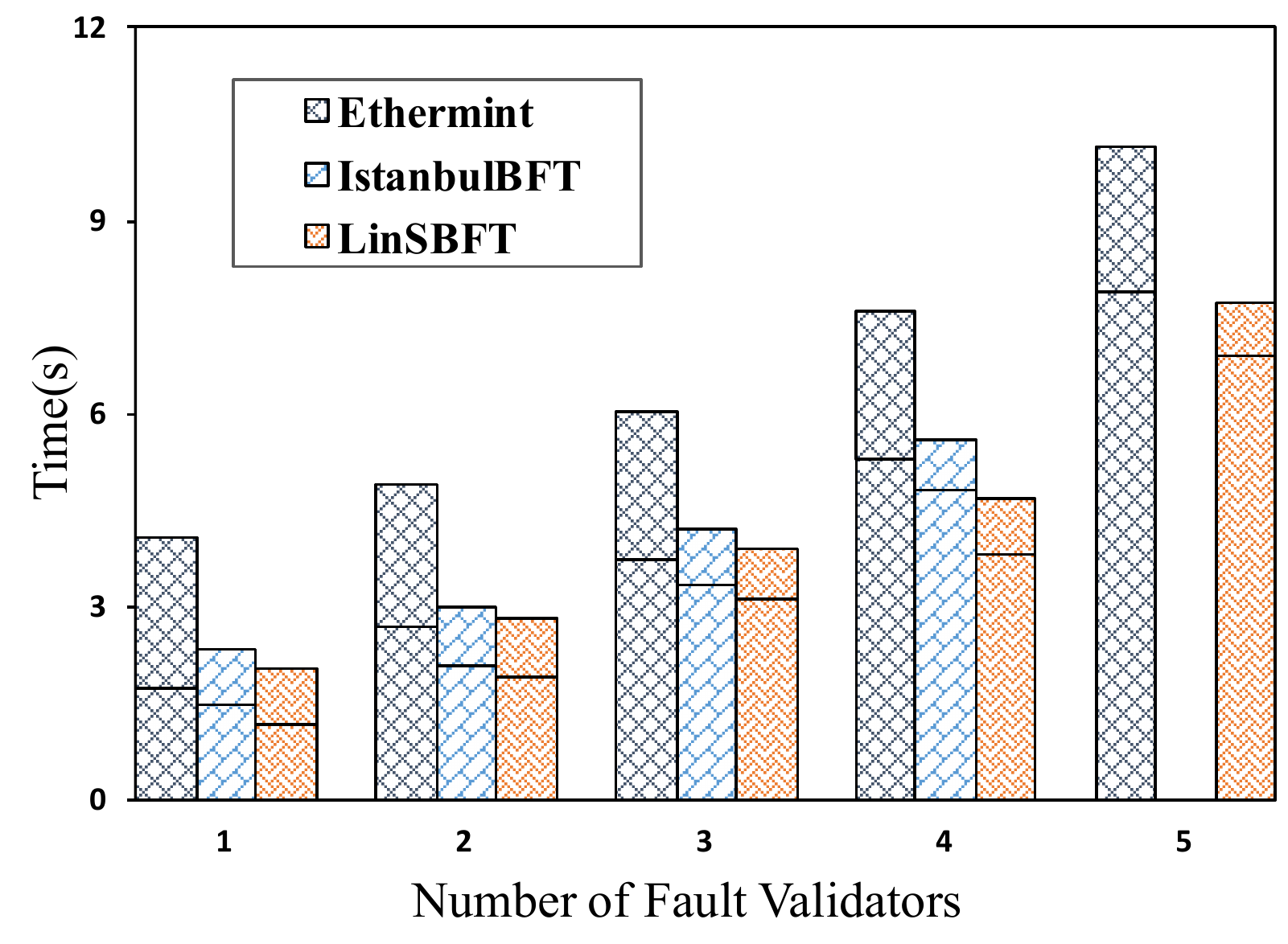}}
	\subfigure[$TXs=2000, \lambda=3s$, crash fault]{\label{subfigure:fault_2000_3s_time}
		\includegraphics[width=1.65in,height=1.3in]{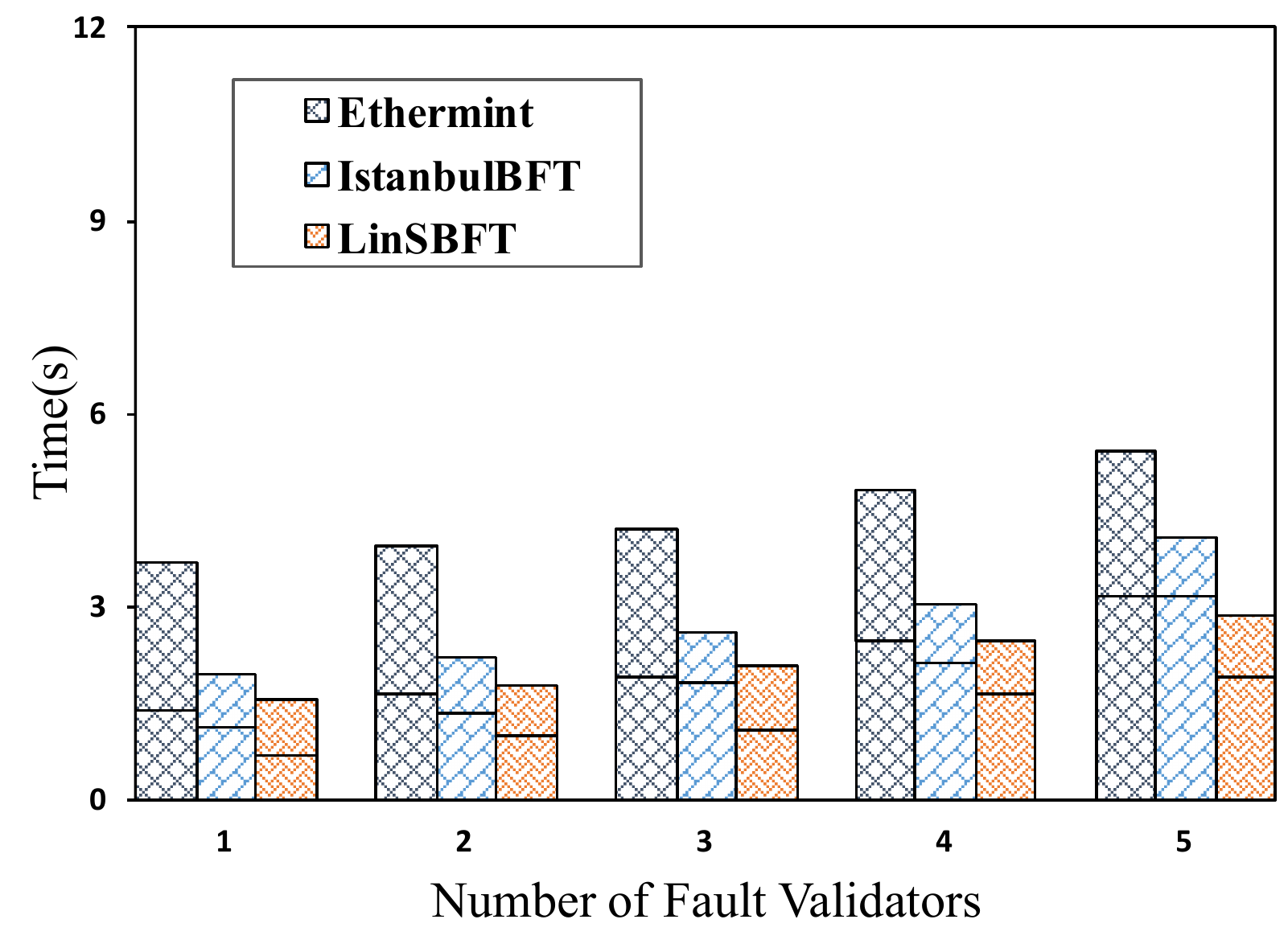}}
	\subfigure[$TXs=2000, \lambda=10s$, byzantine fault]{\label{subfigure:fault1_2000_10s}
		\includegraphics[width=1.65in,height=1.3in]{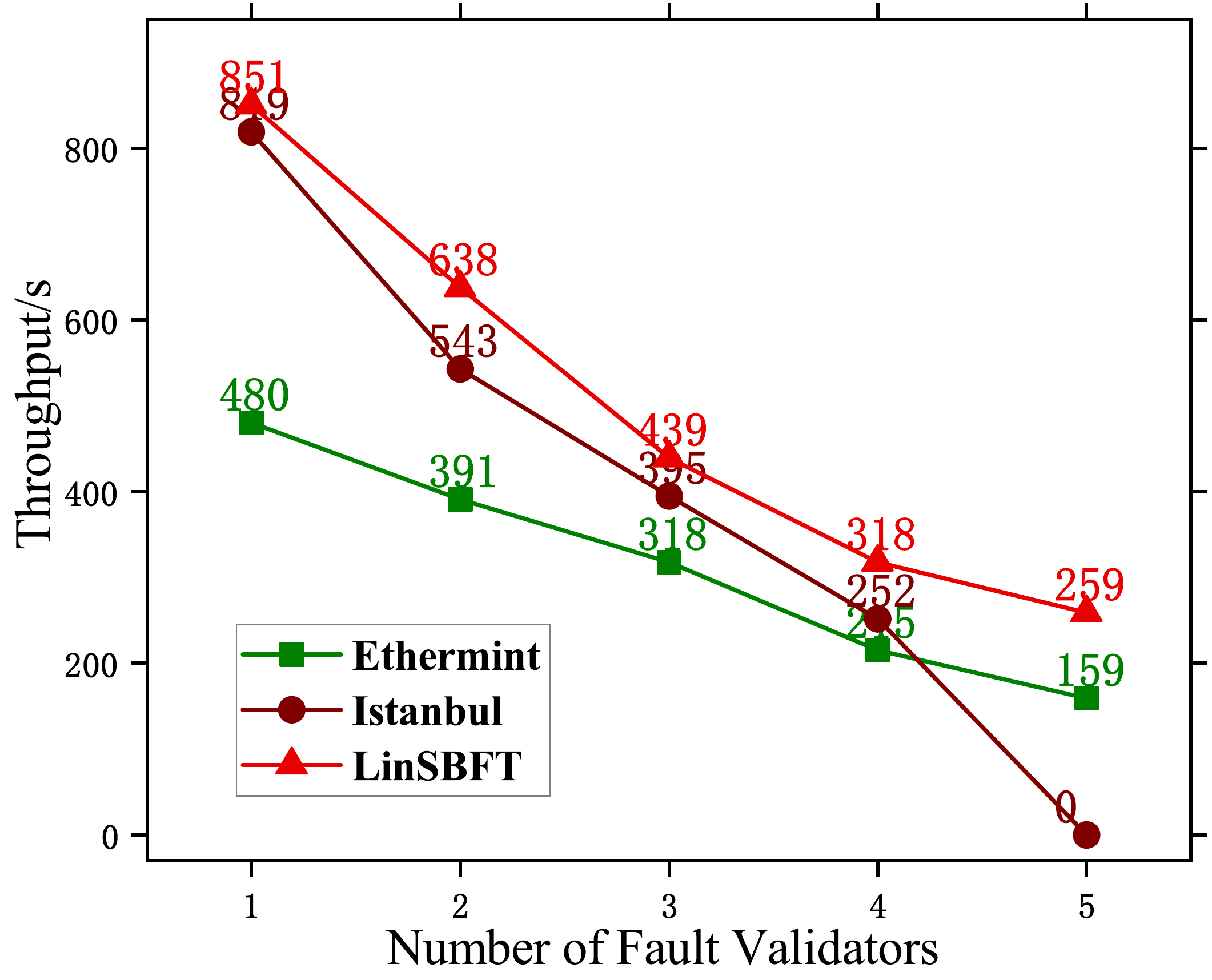}}
	\subfigure[$TXs=2000, \lambda=3s$, byzantine fault]{\label{subfigure:fault1_2000_3s}
		\includegraphics[width=1.65in,height=1.3in]{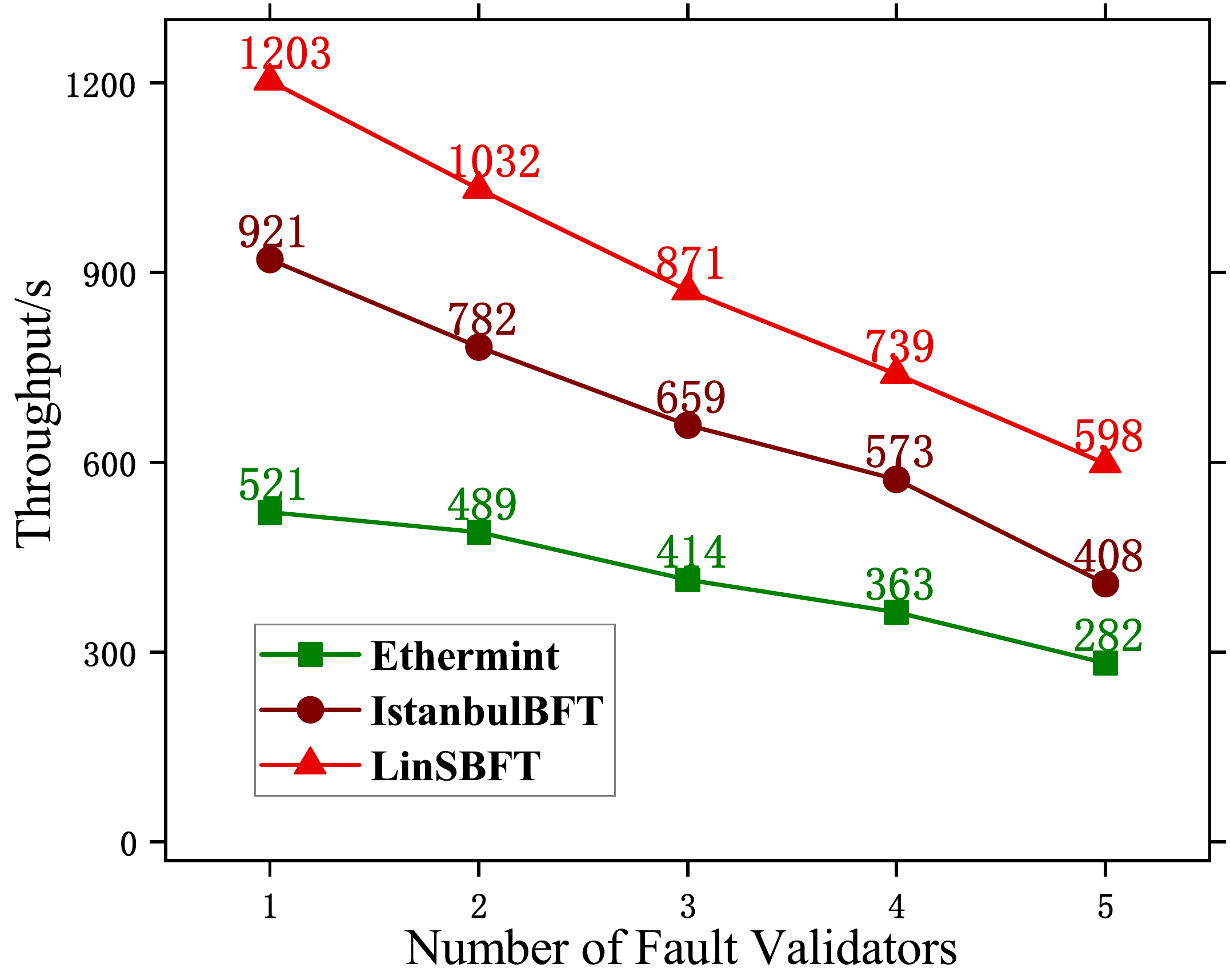}}
	\subfigure[$TXs=4000, \lambda=3s$, byzantine fault]{\label{subfigure:fault1_4000_3s}
		\includegraphics[width=1.65in,height=1.3in]{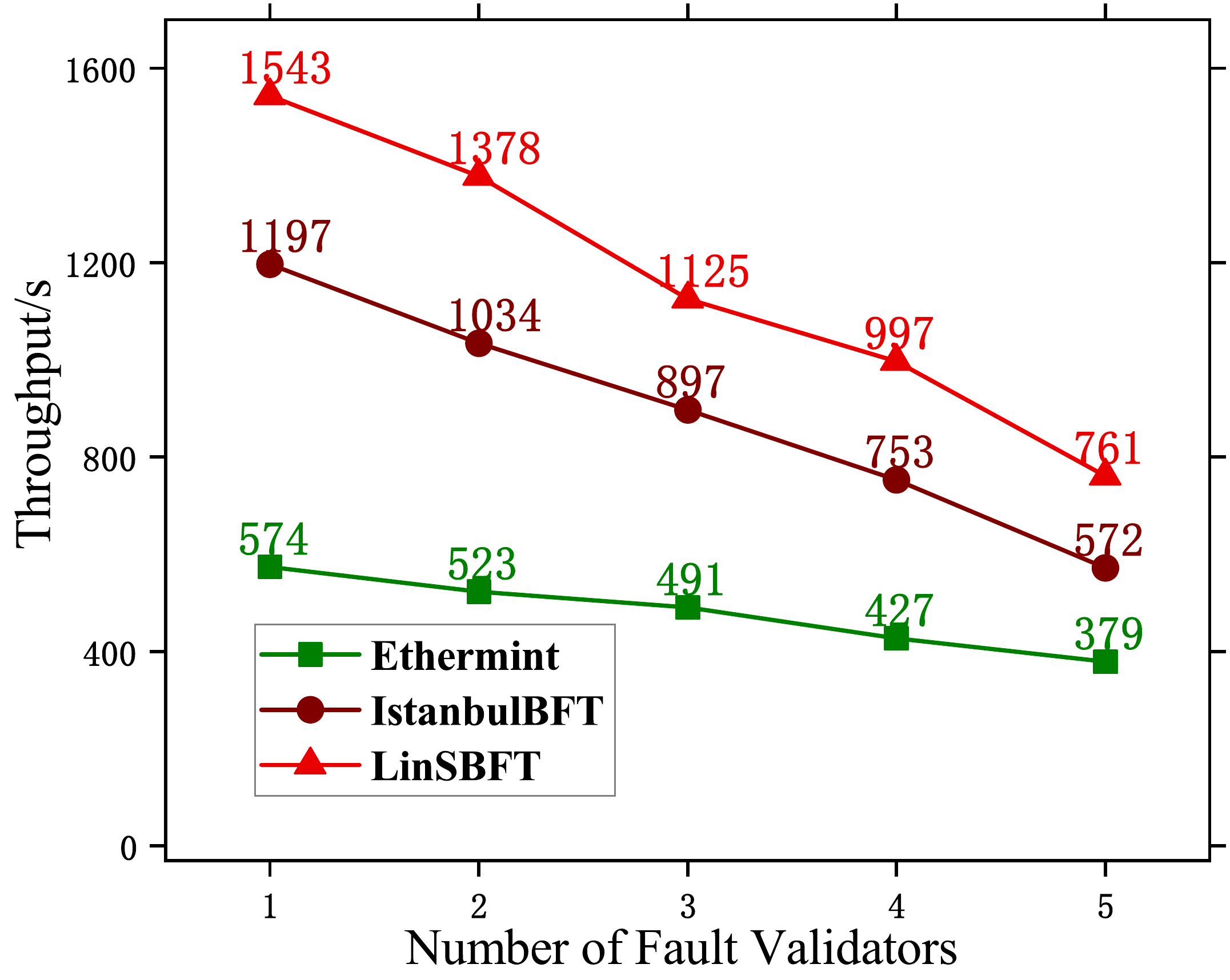}}
	\caption{\label{figure:fault} Throughput and latency with faulty validators.}
\end{figure*}

We implemented LinSBFT in the PChain platform\footnote{https://www.pchain.org/cn}, with open APIs to testers and developers. PChain is backward compatible with all Ethereum Virtual Machine (EVM) instructions. Hence, smart contract transactions for Ethereum can be executed in PChain without modification. Digital signature and VRF are implemented using the BN256 elliptic curve\cite{bn256}, and the hash function is Keccak-256 as in Ethereum. Each validator maintains a TCP connection with its peers, and all validators are reachable with each other via the P2P network. The PChain TestNet is deployed in a cluster of virtual machine instances on Amazon EC2, scattered in several geographic regions including North America, Europe, and Asia. In our experiments, each node is equipped with 16 CPU cores running at 2.10GHz, 96GB RAM, 3TB disk space and up to 1 Gbps network throughput (throttled in our experiments). We compare our system to IstanbulBFT\footnote{https://github.com/ethereum/EIPs/issues/650} and Ethermint, which are also in-production BFT protocols for Ethereum-compatible systems. We do not compare to research prototypes that have not been tested in production, since implementation details can affect performance significantly, as we show below.

We vary the number of validators from $2$ to $64$. Each machine continuously generates transactions with multiple threads, and these transactions are gossiped to all validators. We employ a dataset collected from the PChain TestNet in our experiments. To ensure that the consensus overhead dominates overall latency rather than CPU time consumed by the EVM, we limit the experiments to token transfer smart contract transactions.


We evaluate the LinSBFT on two metrics: throughput and transaction latency in terms of total wallclock time. In order to conduct a comprehensive evaluation, we also simulated validators with no fault, fail-stop faults and byzantine faults, respectively. In the no-fault setting, we test throughput and wallclock time for all systems, with number of transactions ($TXs$) in a block varying from 2000 to 8000. We also deployed all systems in a limited-bandwidth setting to show the advantages of LinSBFT's lower overall transmission volume. In the setting with faulty validators, we test the throughput and response time of all BFT systems with different timeout parameter $\lambda$; a timeout triggers the view change subprotocol necessary for ensuring liveness. Regarding the Byzantine fault setting, it is infeasible to simulate all possible attack strategies to test the safety of LinSBFT (guaranteed by Claim \ref{claim:safety}). Our experiments consider the following strategies: remaining silent (called \emph{crash fault}) and sending different messages to a subset of validators (\emph{Byzantine fault}), and demonstrate LinSBFT's robustness in the presence of such malicious validators.

\subsection{Performance with No Fault}

\textbf{Throughput}. 
Figs.\ref{subfigure:ordinary_tps_2000}-\ref{subfigure:ordinary_tps_8000} illustrate the throughput for the three BFT systems in the case where there is no faulty validator. To avoid view change, we set the timeout parameter $\lambda$ to a large value for all systems. Compared to IstanbulBFT and Ethermint, LinSBFT obtains consistently higher throughput in terms of transactions per second (TPS). The reason is that LinSBFT needs only a single phase of voting to achieve consensus. Besides, although the network condition is favorable (1Gbps bandwidth), the use of threshold signature in LinSBFT still improves performance, as it avoids broadcasting and collecting $O(n)$ voting messages at each node, as is done in IstanbulBFT and Ethermint. The performance gain of LinSBFT is more pronounced as the number of validators increases.

As expected, the throughput for all systems increase with the number of transactions per block. It is worth noting that when the size of participant set is small $(<3)$, each validator proposes the number of transactions is less than the size of block because verification of signatures consumes most of CPU time. The performance of Ethermint is significantly lower than the other two systems due to its protocol implementation details, which is discussed in the following latency analysis.

\textbf{Latency}. Figs. \ref{subfigure:ordinary_time_2000}-\ref{subfigure:ordinary_time_8000} show the latency for all three systems. In Figure \ref{figure:time_cost}, the top of each bar represents execution time for transactions and the bottom represents the time of reaching consensus per block. Figs. \ref{subfigure:ordinary_time_2000}-\ref{subfigure:ordinary_time_8000} illustrate that Ethermint has the highest time overhead for transaction execution, while IstanbulBFT and LinSBFT incur similar transaction execution time. The reason is that when finalizing a block, Ethermint needs to send an extra \emph{deliver\_tx} RPC request for each transaction to Ethereum for execution.  In addition, Ethermint needs to recheck all transactions in Tendermint's \emph{mempool} by sending \emph{check\_tx} RPC requests to Ethereum after the finalization of a block to delete invalid transactions (e.g., double spending ones). For the sake of fairness, we turn off the recheck function of Ethermint.  We have verified that the execution overhead of block dominates time cost in all systems, when the number of validators is small.  As the number of participants increases, the consensus time overhead increases because more validators lead to more time to broadcast messages, especially in IstanbulBFT and Ethermint. Under favorable network conditions, in general the bottleneck lies not in network traffic of consensus, but in transaction execution, as shown in Figure \ref{subfigure:ordinary_time_8000}. Nevertheless, LinSBFT only consumes half of the time of IstanbulBFT, thanks to its single-phase voting design and linear-communication achieved via threshold signatures. Ethermint takes more time for consensus than IstanbulBFT and LinSBFT, since the former's underlying Tendermint software partitions a block into parts with size of 64KB and broadcasts each parts concurrently. Similar to throughput of three systems, with the number of transactions in each block increasing, the time cost of three systems increases correspondingly.

\textbf{Limited bandwidth experiments}. To determine the effectiveness of LinSBFT's reduction in communication overhead, we throttle the bandwidth for each validator to 1Mbps and 8Mbps with a fixed block size of $TXs=4000$. Figs. \ref{subfigure:ordinary_4000_1Mbps} and \ref{subfigure:ordinary_4000_1Mbps_time} show the TPS and total time with 1Mbps bandwidth. The throughput of all systems decreases quickly with increasing number of validators, which indicate that communication cost has become the bottleneck. Note that when the number of validators reaches $64$ in Figure \ref{subfigure:ordinary_4000_1Mbps}, Ethermint and IstanbulBFT can no longer obtain consensus, despite their theoretical liveness guarantees. 

In terms of throughput, the performance gap between LinSBFT and  IstanbulBFT/Ethermint increases with the number of validators, due to their linear and quadratic transmission volume, respectively.  With $8$Mbps bandwidth,  LinSBFT achieves $497$ and $173$ TPS  with $32$ and $64$ participants  while Istanbul only has $132$ and $28$ TPS respectively. Results about 8Mbps are omitted due to space limitations. Results on total wallclock time, shown in Figure \ref{subfigure:ordinary_4000_1Mbps_time}, lead to similar conclusions, with LinSBFT consistently and significantly outperforming its competitors.

\subsection{Performance with Faulty Nodes}
\textbf{Crash faults}. We first evaluate the throughput and latency of the three systems with crash faults. The number of validators is fixed  to $16$, meaning that a BFT protocol can tolerant up to $5$ faulty validators. Figs. \ref{subfigure:fault_2000_10s}-\ref{subfigure:fault1_4000_3s} show the results with the number of faulty validators varying from $1$ to $5$. In Figs. \ref{subfigure:fault_2000_3s} and \ref{subfigure:fault_4000_3s}, the block size is $2000$ and $4000$ transactions, respectively. In Figs. \ref{subfigure:fault_2000_10s} and \ref{subfigure:fault_2000_3s}, timeout parameter $\lambda$ is set to $10s$ and $3s$, respectively. Observe that with increasing number of faulty validators,  throughput generally decreases since crashed validators lead to timeouts and expensive view changes. More faulty validators lead to more frequent view changes, and thus, larger costs. With a high $\lambda$, LinSBFT and IstanbulBFT demonstrate similar performance. When $\lambda$ becomes lower as shown in Figs. \ref{subfigure:fault_2000_3s} and \ref{subfigure:fault_4000_3s},  LinSBFT has significant advantage over IstanbulBFT and Ethermint. Similar to the ordinary case, a larger block size results in higher performance with fault validators. Notes that validators in IstanbulBFT can no longer achieve consensus in the setting with $5$ fault validators and $\lambda=10$s.

Figs. \ref{subfigure:fault_2000_10s_time}-\ref{subfigure:fault_2000_3s_time} show that total running time of all systems increases with the number of faulty validators. We measure the average time cost for transaction execution and consensus per block. As we expected, a larger $\lambda$ results in higher wall clock time for consensus. Notably, when $\lambda$ is relatively small, LinSBFT achieves 1.5x performance boost compared to IstanbulBFT, in terms of TPS.

\textbf{Byzantine faults}. Finally, we evaluate the performance of all systems in the presence of Byzantine faults. Figs. \ref{subfigure:fault1_2000_10s}-\ref{subfigure:fault1_4000_3s} illustrate that throughput of all systems decreases with increasing number of faulty validators, as expected. Compared to the results with crash faults, the performance of all three systems has declined to varying degrees due to the locking mechanism of systems. 
Specifically, we assume that when a malicious validator becomes a leader (collector), it only sends messages to half of the honest validators, who may be locked on the proposal until a locked validator becomes the leader. Compared to Ethermint and IstanbulBFT, the performance degradation of LinSBFT is the lowest, since the \emph{Propose-Lock} makes a successful collector propose the locked block that the locked validators may vote for. However, with consecutive faulty collectors, LinSBFT incurs increased overhead due to locking, and requires more rounds to change to an honest collector. Due to space limitations, results on wallclock time are omitted, which lead to similar conclusions as those for crash fault. In practice, we expect the proportion of faulty validators to be low, since such behavior can be disincentivized via the block reward mechanism, which is outside the scope of this paper.

%% file: relatedwork.tex
\section{Related work}\label{sec:rw}

BFT protocols have been extensively studied in the traditional distributed systems setting. An early influential work is the DLS protocol \cite{DBLP:journals/jacm/DworkLS88}, which achieves safety and liveness, at the expense of prohibitive $O(n^4)$ communication cost. Castro and Liskov propose PBFT \cite{DBLP:conf/osdi/PBFT}, which incurs $O(n^2)$ transmissions in the ordinary case. As pointed out in \cite{DBLP:journals/corr/abs-1803-05069}, this is essentially an optimistic run, and the protocol falls back to DLS when the optimistic run fails. Later work, e.g., Zyzzyva \cite{DBLP:journals/tocs/Zyzzyva}, further improves the efficiency of the optimistic run.

In traditional BFT researches \cite{DBLP:conf/osdi/PBFT, DBLP:journals/tocs/Zyzzyva, liu2018scalable} discussed above, it is commonly assumed that there is a fixed cluster of verifier nodes. Meanwhile, in many protocols, the same node stays as the leader unless a view change occurs. Further, an honest node is assumed to stay honest, regardless of the number of transaction batches it verifies. These assumptions are unrealistic in a public blockchain setting. Tendermint \cite{buchman2016tendermint}, based on PBFT, runs consensus for each block with a rotating leader scheme, which is more suitable for blockchains. However, Tendermint still incurs $O(n^3)$ worst-case communication volume, and its adaptation Ethermint has been shown to perform poorly in our experiments.

Recently, Casper\cite{DBLP:journals/Casper} amortizes the cost of its BFT protocol by running consensus once for multiple (100 in \cite{DBLP:journals/Casper}) block heights. This design, however, gives much power to the block proposer. Hence, Casper involves a PoW mechanism for leader selection, which runs the risk of forks and 51\% attacks. Hot-Stuff \cite{DBLP:journals/corr/abs-1803-05069} improves the worst case communication complexity to $O(n^2)$, using a combination of linear view change and threshold signatures. SBFT \cite{DBLP:journals/corr/SBFT} reduces the communication complexity of the ordinary case using threshold signatures and collectors. Omniledger \cite{DBLP:conf/sp/Kokoris-KogiasJ18} achieves $O(\log n)$ time in the best case, with the help of the CoSi protocol \cite{GOOGLE:syta2016keeping}. Several protocols distinguish malicious nodes who actively attack the protocol with falsified messages from ones that may fail-top (e.g., in \cite{DBLP:journals/corr/SBFT}) or those that may go offline \cite{DBLP:conf/asiacrypt/PassS17}, and obtain stronger robustness in a setting where only a small fraction of nodes are actively malicious. None of these protocols, however, have been deployed in production for a large-scale blockchain network.

Another promising trend is BFT protocols with probabilistic guarantees on safety and liveness. Dfinity \cite{GOOGLE:Dfinity} uses a random sample set of nodes to verify a block. Since its safety guarantee is probabilistic, the sample set needs to be sufficiently large to obtain a low probability of failure. Hence, it still needs a scalable, deterministic BFT sub-protocol for the sample set. Algorand \cite{DBLP:conf/sosp/Alogrand} addresses the situation where the adversary is adaptive, who can instantly corrupt any node at will. As pointed out by Chan et al. \cite{DBLP:journals/corr/CEBA}, Algorand replies on a public key infrastructure, which may not exist in a public blockchain. LinSBFT does not consider an adaptive adversary since (\romannumeral1) for a fast protocol, compromising validators adaptively within a round is rather difficult, and (\romannumeral2) the only part of the protocol that is vulnerable to an adaptive dversary is random leader selection, for which the adversary can break the probabilistic guarantee on $O(1)$ leader rotations by corrupting a considerable portion of the nodes; this is difficult, however, for a larger $n$.

Finally, consensus protocols based on Directed Acyclic Graphs
(DAG)\cite{dagcoin,byteball,hashgraph} seek consensus on individual transactions rather than blocks, and confirm them concurrently by expanding a hash-linked graph of transactions. Such approaches, however, generally do not guarantee bounded latency, as a new transaction can wait indefinitely until another chooses to confirm it.

%% file: conclusion.tex
\section{Conclusions} \label{sec:conclusion}
The paper proposes LinSBFT, an in-production BFT protocol that achieves amortized $O(n)$ communication cost, requires only a single phase of voting in the ordinary case, satisfies deterministic guarantees on safety and liveness, and is suitable for  a public, permissionless blockchain setting with a dynamic validator set and potentially changing honesty. Experiments with real data demonstrate the advantages of LinSBFT in terms of throughput and latency, under various assumptions of node faults. Regarding future work, an interesting direction is to investigate the combination of LinSBFT with a randomized BFT protocol, as well as other scaling options such as sharding \cite{DBLP:conf/sp/Kokoris-KogiasJ18}.

%% file: appendix.tex
\begin{appendices}
\section{}
\subsection{Future Proposals and Votes} \label{append:future}

\noindent\noindent \textbf{Lemma} 4. \emph{Unless with the negligible probability, the number of validators who broadcast proposal is constant.}

\begin{proof}
	First, we proof that, in expectation,  an honest validator takes $4/3$ view changes to a round with honest collector. Assume it takes $I$ view changes for a validator to enter a round whose collector is honest. The collector is malicious for each round with a  probability $p=f/n$, and probability that $I$ equals $i$ is $Pr{\left[I=i\right]}=p^{i-1}\left(1-p\right)$. The expectation of $I$ is calculated as follow:
	
	{\footnotesize
	\begin{flalign*}
	E(I)&= \sum\limits_{i=1}^{\infty}{iPr{\left[I=i\right]}} =\left(1-p\right)\sum\limits_{i=1}^{\infty}{ip^{i-1}}\\
	&=\left(1-p\right)\left(\sum\limits_{i=1}^{\infty}{ix^{i-1}}\right)\bigg|_{x=p}\\
	&=\left(1-p\right)\left(\sum\limits_{i=1}^{\infty}{\left(x^{i}\right)'}\right)\bigg|_{x=p}=\left(1-p\right)\left(\sum\limits_{i=1}^{\infty}{x^{i}}\right)'\bigg|_{x=p}\\
	&=\left(1-p\right)\left({\frac{x}{1-x}}\right)'\bigg|_{x=p}=\left(1-p\right)\left({1-x}\right)^{-2}\bigg|_{x=p}\\
	&=\left({1-p}\right)^{-1}<\left(1-1/4\right)^{-1}=\frac{4}{3}
	\end{flalign*}
	}
	Second, the probability that the next collector is malicious becomes negligible after a constant number($I>-\log_{4}\rho$) through random leader selection. In practice, number $I$ floats around $E(I)$ in the vast majority of cases. Hence, the consensus for each height is terminated after constant number of rounds. 
	
	Last, if malicious validators want more honest validators to broadcast messages, they have to send future votes to honest collector. However, the calculated collectors may be fault with a probability $p$ as well. It makes no sense to send future vote to a fault validator. In order to win more time, malicious validator just remain silent. During the time for view change, valid future votes are produced. If there are $I$ consecutive malicious collectors, at most $I$ valid votes are produce and expectation of $I$ is not greater than $4/3$. When an honest validator becomes the collector runs at largest height, it propose a block with its largest $Cert$ that can help delayed validator catch up. In addition, the honest validator sends its vote with largest $Cert$ to corresponding collector. An honest collector for future round broadcasts a $Cert$ message to notify others upon receiving a vote certificate succeeding itself. Hence,  with the negligible probability, the number of validators who broadcast proposal is constant. 
\end{proof}

\subsection{Correctness for Liveness} \label{append:liveness}
In this part, we presents the proof of Claim \ref{claim:liveness} described in Section \ref{subsec:correctness} via three lemmas.
\begin{lemma}\label{lemma:height}
	Let $xT$ be a time after GST. It supposes that $l$ is the maximum height number of all honest validators. Then by $xT+\Delta$, all honest validators at least move to height $l-1$. 
\end{lemma}

\begin{proof}
	The honest validators running at height $l$ broadcast \emph{State} messages with largest vote certificates for height $l-1$. The validators staying height $l'$($l'\le l-2$) synchronize data and move to height $l$ upon receiving these \emph{State} messages with vote certificates succeeding themselves. Therefore, all honest validators run at height not less than $l-1$.
\end{proof}

\begin{lemma} \label{lemma:sync}
	Let $xT$ be a time after GST. It supposes that $l$ is the largest height number of all honest validators. Then all honest validators must move to height $l$ eventually. 
\end{lemma}

\begin{proof}
	If all honest validators run at height $l$, the lemma is held. Otherwise, there must be valiators stay at height less than $l$. According to Lemma \ref{lemma:height}, all honest validators move to at least height $l-1$ by $xT+\Delta$. The validators running at height $l-1$ will move to higher height after synchronization upon receiving certificate succeeding itself. In addition, all validators participant in consensus for height $l-1$ even some of them run at height $l$. Due to the doubling of \emph{Timeout}, all validators enter same round after a round number $r$ at height $l-1$ forever. With negligible probability, there must be honest collectors for two successive rounds, say $r', r'+1$ such that $r'>r$. In round ${R}_{l-1,r'}$, the certificate $Cert$ owned by honest validator for a round at height $l-1$ that can vote-unlocks and propose-unlocks all other honest validators must be sent to collector $C_{l-1,r'}$. If $C_{l-1,r'}$ can not derive a vote certificate, it proposes the block matched $Cert$ that will be accepted by all validators. Hence, collector $C_{l-1,r'+1}$ can produce a valid vote certificate which let all honest validators move to height $l$. If we set multiple collectors, as long as one of $k$ collectors for round ${R}_{l-1,r'+1}$ is the honest, validators move to height $l$ and the lemma is held. 
\end{proof}

\begin{lemma} \label{lemma:move}
	It supposes that $l$ is the highest height of all honest validators. At least one honest validator must move to height $l+1$ eventually. 
\end{lemma}

\begin{proof}
	If there are more than $n-f$ validators running at height $l$, and a vote certificate $Cert$ is derived for round at height $l$. When $Cert$ is broadcast by anyone,  all honest validators move to height $l+1$ at once. Otherwise, according to Lemma \ref{lemma:sync}, all honest validators must move to height $l$. We suppose the $r$ is the maximum round number of all honest validators when they all move to height $l$. Then all honest validator will enter  same round  after round $r$ forever due to the doubling of timeout. After round ${R}_{l,r}$, as long as the collector is honest, there must be one honest validator that moves to height $l+1$. Besides, there is at least one honest validator moving to height $l+1$ if the malicious collector partitions honest validators deliberately by sending message to part of honest validators. Hence, the lemma is held.
\end{proof}

According to the proofs for above three lemmas, there must be honest validator moving to larger height forever, therefore the consensus for each block height must  terminate within finite time in a partially synchronous network. 
\end{appendices}